%% file: RobustMIMO.tex
\documentclass[onecolumn,draftclsnofoot,nofonttune,10pt]{IEEEtran}

\usepackage{amsmath}
\usepackage{amssymb}
\usepackage{amsfonts}
\usepackage{amsthm}

\usepackage[utf8]{inputenc}
\usepackage[T1]{fontenc}
\usepackage[varg]{txfonts}
\let\mathbb=\varmathbb

\usepackage{dsfont}

\usepackage[dvipsnames,svgnames]{xcolor}
\colorlet{MyBlue}{DodgerBlue!40!Black}
\colorlet{MyGreen}{DarkGreen!50!Black}

\usepackage{subfigure}

\usepackage{acronym}
\usepackage{balance}
\usepackage{bbm}
\usepackage{latexsym}
\usepackage{paralist}
\usepackage{wasysym}
\usepackage{xspace}

\usepackage[numbers,sort&compress]{natbib}

\usepackage{hyperref}
\hypersetup{
draft=false,
colorlinks=true,
linktocpage=true,
pdfstartview=FitH,
breaklinks=true,
pdfpagemode=UseNone,
pageanchor=true,
pdfpagemode=UseOutlines,
plainpages=false,
bookmarksnumbered,
bookmarksopen=false,
bookmarksopenlevel=1,
hypertexnames=true,
pdfhighlight=/O,
urlcolor=Maroon,linkcolor=MyBlue,citecolor=MyGreen, 
pdftitle={},
pdfauthor={},
pdfsubject={},
pdfkeywords={},
pdfcreator={pdfLaTeX},
pdfproducer={LaTeX with hyperref}
}

\newcommand{\bH}{\mathbf{H}}
\newcommand{\bI}{\mathbf{I}}

\newcommand{\bP}{\mathbf{P}}
\newcommand{\bQ}{\mathbf{Q}}

\newcommand{\bV}{\mathbf{V}}
\newcommand{\bW}{\mathbf{W}}

\newcommand{\bY}{\mathbf{Y}}
\newcommand{\bZ}{\mathbf{Z}}

\newcommand{\bu}{\mathbf{u}}

\newcommand{\bx}{\mathbf{x}}
\newcommand{\by}{\mathbf{y}}
\newcommand{\bz}{\mathbf{z}}

\newcommand{\C}{\mathbb{C}}

\newcommand{\N}{\mathbb{N}}

\DeclareMathOperator*{\argmax}{arg\,max}

\DeclareMathOperator{\bigoh}{\mathcal O}

\DeclareMathOperator{\ex}{\mathbb{E}}

\DeclareMathOperator{\prob}{\mathbb{P}}

\DeclareMathOperator{\tr}{tr}
\DeclareMathOperator{\var}{Var}

\newcommand{\dd}{\:d}

\newcommand{\eps}{\varepsilon}

\newcommand{\mgeq}{\succcurlyeq}

\newcommand{\wilde}{\widetilde}

\newcommand{\abs}[1]{\left\lvert #1 \right\rvert}

\newcommand{\norm}[1]{\left\| #1 \right\|}
\newcommand{\smallnorm}[1]{\| #1 \|}

\newcommand{\given}{\:\vert\:}
\newcommand{\vargiven}{\:\middle\vert\:}

\newcommand{\dis}{\displaystyle}
\newcommand{\txs}{\textstyle}

\newcommand{\insum}{\sum\nolimits}

\usepackage[textwidth=30mm]{todonotes}
\usepackage{soul}
\setstcolor{red}
\sethlcolor{SkyBlue}

\usepackage{algorithm2e_extension}
\usepackage{algorithmic, algorithm}
\SetKwBlock{Repeat}{Repeat}{}
\newcommand{\kwd}[1]{\textsf{\bfseries#1}}

\newtheorem{theorem}{Theorem}

\newtheorem*{corollary*}{Corollary}
\newtheorem{lemma}{Lemma}
\newtheorem{proposition}{Proposition}

\theoremstyle{definition}

\newtheorem*{definition*}{Definition}

\theoremstyle{remark}
\newtheorem{remark}{Remark}
\newtheorem*{remark*}{Remark}

\newcommand{\one}{\mathds{1}}

\newcommand{\step}{\gamma}
\newcommand{\erg}{\textup{erg}}
\newcommand{\rate}{R}
\newcommand{\ergrate}{\rate_{\erg}}
\newcommand{\noisedev}{\sigma}

\newcommand{\pay}{R}

\newcommand{\eq}{\bQ^{\ast}}
\newcommand{\eqset}{\spectron^{\ast}}

\newcommand{\fench}{F}

\newcommand{\hypref}[1]{\textup{(\hyperref[#1]{H\ref*{#1}})}}

\newcommand{\strat}{\mathcal{X}}

\newcommand{\play}{\mathcal{K}}

\newcommand{\spectron}{\mathcal{Q}}

\begin{document}


\title{In an Uncertain World: Distributed Optimization in MIMO Systems with Imperfect Information}

\author{%
Panayotis Mertikopoulos%
,~\IEEEmembership{Member,~IEEE},
and
Aris L. Moustakas%
,~\IEEEmembership{Senior Member,~IEEE}
\thanks{%
This research was supported in part by the European Commission in the framework of the FP7 Network of Excellence in Wireless COMmunications NEWCOM\# (contract no. 318306),
and by the French National Research Agency projects
NETLEARN (ANR\textendash 13\textendash INFR\textendash 004)
and GAGA (ANR\textendash 13\textendash JS01\textendash 0004\textendash 01).
Part of this work was presented in ISIT 2012 and ISIT 2014 \cite{MBM12,CGM14b}.}
\thanks{P.~Mertikopoulos is with the French National Center for Scientific Research (CNRS) and the Laboratoire d'Informatique de Grenoble, Grenoble, France;
A.~L.~Moustakas is with the Physics Department, University of Athens, Greece and Supélec, Gif-sur-Yvette, France, supported by the Digiteo Senior Chair "ASAPGONE". }
} 

\maketitle

\newacro{AIMD}{additive increase, multiplicative decrease}
\newacro{5G}{fifth generation}
\newacro{SISO}{single-input and single-output}
\newacro{MIMO}{multiple-input and multiple-output}
\newacro{MUI}{multi-user interference-plus-noise}
\newacro{MAC}{multiple access channel}
\newacro{PMAC}{parallel multiple access channel}
\newacro{CSI}{channel state information}
\newacro{CSIT}{channel state information at the transmitter}
\newacro{BS}{base station}
\newacro{TDD}{time-division duplexing}
\newacro{CDMA}{code division multiple access}
\newacro{FDMA}{frequency division multiple access}
\newacro{DSL}{digital subscriber line}
\newacro{SIC}{successive interference cancellation}
\newacro{SUD}{single user decoding}
\newacro{SINR}{signal-to-interference-and-noise ratio}
\newacro{KKT}{Ka\-rush--Kuhn--Tuc\-ker}
\newacro{WF}{water-filling}
\newacro{IWF}{iterative water-filling}
\newacro{SWF}{simultaneous water-filling}
\newacro{iid}[i.i.d.]{independent and identically distributed}
\newacro{OFDMA}{orthogonal frequency-division multiple access}
\newacro{MXL}{matrix exponential learning}
\newacro{AMXL}[MXL-a]{asynchronous matrix exponential learning}
\newacro{EXL}[MXL-eig]{eigen-based exponential learning}
\newacro{FCC}{Federal Communications Commission}
\newacro{NTIA}{National Telecommunications and Information Administration}
\newacro{GAO}{General Accounting Office}
\newacro{QoE}{quality of experience}
\newacro{QoS}{quality of service}
\newacro{OFDM}{orthogonal frequency division multiplexing}\acused{OFDM}
\newacro{EW}{exponential weight}
\newacro{OGD}{online gradient descent}
\newacro{OMD}{online mirror descent}
\newacro{APT}{asymptotic pseudotrajectory}
\newacro{ICT}{internally chain transitive}
\newacro{MSE}{mean squared error}
\newacro{EPA}{extended pedestrian A}
\newacro{EVA}{extended vehicular A}

\begin{abstract}
\input{Abstract}
\end{abstract}


\acresetall

\section{Introduction}
\label{sec:intro}
\input{Introduction}

\section{System Model}
\label{sec:model}
\input{Model}
\section{Learning with Imperfect and Delayed Information}
\label{sec:learning}
\input{Learning}

\section{The Case of Fast-Fading Channels}
\label{sec:fading}
\input{Fading}

\section{Numerical Simulations}
\label{sec:numerics}
\input{Numerics}

\section{Conclusions and Perspectives}
\label{sec:conclusions}
\input{Conclusion}

\appendix[Technical Proofs]
\label{app:proofs}
\input{App-Proofs}

\bibliographystyle{IEEEtran}
\footnotesize
\setlength{\bibsep}{0pt}
\bibliography{IEEEabrv,Bibliography}

\end{document}

%% file: Abstract.tex
In this paper, we introduce a distributed algorithm that optimizes the Gaussian signal covariance matrices of multi-antenna users transmitting to a common multi-antenna receiver under imperfect and possibly delayed \acl{CSI}.
The algorithm is based on an extension of exponential learning techniques to a semidefinite setting and it requires the same information as distributed \acl{WF} methods.
Unlike \acl{WF} however, the proposed \ac{MXL} algorithm converges to the system's optimum signal covariance profile under very mild conditions on the channel uncertainty statistics;
moreover, the algorithm retains its convergence properties even in the presence of user update asynchronicities, random delays and/or ergodically changing channel conditions.
In particular, by properly tuning the algorithm's learning rate (or step size), the algorithm converges within a few iterations, even for large numbers of users and/or antennas per user.
Our theoretical analysis is complemented by numerical simulations which illustrate the algorithm's robustness and scalability in realistic network conditions.

%% file: Introduction.tex

Following the seminal prediction that the use of multiple antennas in signal transmission and reception can lead to substantial performance gains \cite{FG98,Tel99}, \ac{MIMO} technologies have become an integral component of state-of-the-art wireless communication protocols (such as 3G LTE, 4G, HSPA+ and WiMax to name but a few).
To capitalize on these gains, the emerging massive \ac{MIMO} paradigm, a contending technology for 5G networks, ``goes large'' by scaling up existing multiple-antenna transceivers through the use of inexpensive service antennas and \ac{TDD} in order to focus energy into ever smaller regions of space \cite{HtBD13,RPL+13,LETM14}.
In so doing, massive \ac{MIMO} arrays can increase throughput by a factor of $10\times$ (or more),
bring about significant latency reductions over the air interface,
and greatly improve the system's robustness to ambient noise \cite{LETM14,ABC+14}.

Nevertheless, when coupled with the projected network densification (e.g. due to the massive deployment of small cells), the resulting antenna density is expected to create increased interference which will have to be mitigated through spatial focusing.
To that end, it is crucial to optimize the input signal distribution of each user, especially in the moderate \ac{SINR} regime:
in this way, wireless receivers can achieve higher transmission rates with the same power, thus increasing spatial spectrum reuse and improving their radiated energy efficiency and overall \ac{QoE}.

This optimization is typically achieved via \ac{WF} methods \cite{CV93,YRBC04,SPB08-jsac} that rely on the transmitters having access to accurate \ac{CSI}, including their individual transfer matrices and the \ac{MUI} that they are facing at the receiver.
One way to collect the required \ac{CSI} is to have the \ac{BS} emit reverse-link pilot waveforms that allow each terminal to estimate the corresponding channel responses over a given frequency band.
However, given that the number of responses that must be estimated at each terminal is proportional to the number of antennas at the \ac{BS}, massive \ac{MIMO} systems may require up to a hundred times more pilots than conventional \ac{MIMO} systems \cite{RPL+13};
as such, one of the most popular solutions is to operate in \ac{TDD} mode and rely on uplink-downlink reciprocity for the exchange of \ac{CSI} pilot feedback \cite{LETM14}.

In this context, a major challenge occurs when transmitters are required to optimize their input signal covariance matrices in the presence of imperfect and/or delayed \ac{CSI} \textendash\ e.g. due to the vastly increased impact of pilot contamination in massive \ac{MIMO} systems \cite{Mar10,HtBD13}.
In the absence of perfect \ac{CSIT}, the convergence of \ac{WF} methods is no longer guaranteed because the algorithms' fixed point can be significantly perturbed by erroneous (or obsolete) \ac{CSI}.
As a result, the efficient deployment of massive \ac{MIMO} systems calls for flexible and robust optimization algorithms that are capable of dealing with feedback uncertainty originating from temporal fading, lack of adaptation synchronicity, measurement errors due to noise and pilot contamination, etc.


In this paper, we propose a distributed optimization algorithm based on the method of \ac{MXL} that was recently introduced by the authors of \cite{MBM12,CGM14b}.
Essentially, rather than updating their signal covariance matrices, transmitters update the logarithm of these matrices based on (possibly imperfect) measurements of a matrix analogue of the transmitter's \ac{SINR}.
The benefit of updating the logarithm of a user's covariance matrix is that the algorithm's updates only need to be Hermitian (and not necessarily positive-definite), so it is much easier to respect the problem's semidefiniteness constraints.
Furthermore, in contrast to \ac{WF} methods, the proposed algorithm proceeds by \emph{aggregating} \ac{CSI} feedback over time:
in this way, measurement errors, noise and asynchronicities effectively vanish in the long run thanks to the law of large numbers for martingales.
In particular, the proposed algorithm has the following desirable attributes:
\begin{enumerate}
\addtolength{\itemsep}{2pt}
\item
It is \emph{distributed:}
user updates are based on local information and channel measurements.
\item
It is \emph{robust:}
measurements and observations may be subject to random errors and noise.
\item
It is \emph{stateless:}
users do not need to know the state (or topology) of the system.
\item
It is \emph{reinforcing:}
each user tends to increase his own rate.
\item
It is \emph{decentralized:}
user updates need not be synchronized or otherwise coordinated.
\end{enumerate}


A good paradigm to test the performance and convergence properties of the proposed algorithm is the widely studied vector Gaussian \ac{MAC} \cite{YRBC04}.
This channel model is the \ac{MIMO} equivalent of the \ac{PMAC} and consists of several (and mutually independent) \ac{MIMO} transceivers that are linked to a common multi-antenna receiver.
In this framework, assuming perfect \ac{CSIT}, it is well known that \ac{IWF} converges to the system's optimum transmit profile \cite{YRBC04}, but the algorithm's convergence speed decreases proportionally with the number of transmitting users.
On the other hand, \ac{SWF} is much faster, but it may fail to converge altogether:
as was shown in \cite{MBML12}, the sufficient conditions that guarantee the convergence of \ac{SWF} methods may fail to hold even in simple $2\times2$ \ac{PMAC} systems;
To make matters worse, this situation is exacerbated in the case of imperfect \ac{CSI} where even the convergence of \ac{IWF} is no longer guaranteed;
by contrast, the proposed \ac{MXL} algorithm converges rapidly to the system's optimum transmit profile (even for large numbers of users and/or antennas per user), and it remains convergent irrespective of the magnitude of the measurement noise.

Our theoretical analysis relies on the powerful stochastic approximation methods of \cite{Ben99,Bor08} and the matrix regularization machinery of \cite{TRW05,KSST12}.
Specifically, we first establish the algorithm's convergence in a continuous-time setting, and we then use the ODE method of stochastic approximation \cite{Ben99} and the theory of concentration inequalities for martingales \cite{HH80,dlP99} to show that this convergence is retained in a discrete-time setting \textendash\ even under noisy and delayed/asynchronous updates.%
\footnote{For a related continuous-to-discrete descent, see the recent paper \cite{MB14} on unilateral online optimization in dynamically varying cognitive radio systems.}
The algorithm's main tunable parameter is its step-size which controls the rate at which users learn:
picking larger step sizes accelerates the algorithm, but this acceleration comes at the expense of accuracy (a crucial trade-off in the presence of noise).
Still, if the error process has finite $p$-th moments for some $p\geq20$, convergence is guaranteed as long as the algorithm's step-size $\step_{n}$ satisfies $\sum_{n} \step_{n}^{1+p/2}<\infty$;
in practice, the central limit theorem guarantees that estimators of Gaussian variables have finite moments of every order, so the algorithm's step-size can be taken nearly constant, guaranteeing in this way relatively rapid convergence even in very noisy environments.
In addition, we show that, under mild conditions on the moments of the error process, the tail behavior of the deviations is Gaussian \textendash\ i.e. it is sharply concentrated around its mean.

\subsection*{Paper outline and summary of results}

After introducing our system model in the next section, we derive the proposed \acl{MXL} algorithm in Section \ref{sec:MXL} (cf. Algorithm~\ref{alg:MXL});
in the same section, we also state our main convergence result in the presence of measurement noise and errors (Theorem \ref{thm:conv}).
Since one of the main goals of this paper is to illustrate the convergence properties of \ac{MXL} in the presence of various stochastic impediments, Section \ref{sec:estimate} provides a concrete example of feedback matrix estimation, while Section \ref{sec:AMXL} extends the convergence results of Section \ref{sec:MXL} to the case of asynchronous user updates (Theorem \ref{thm:conv-AMXL}).
Section \ref{sec:EXL} describes the evolution of the eigenvalues and eigenvectors of the users' signal covariance matrices (Theorem \ref{thm:conv-EXL}), while Section \ref{sec:fading} extends our results further to the case of fast-fading channels (Theorem \ref{thm:conv-erg}).
Finally, our theoretical analysis is validated and supplemented by numerical simulations in Section \ref{sec:numerics}.

To streamline the flow of the paper, proofs and technical details have been delegated to a series of appendices at the end.

%% file: Model.tex

%
%

Consider a Gaussian vector \acf{MAC} where a finite set of wireless users $k\in\play\equiv \{1,\dotsc,K\}$ transmit simultaneously over a common channel to a base receiver with $N$ antennas.
If the $k$-th transmitter is equipped with $M_{k}$ transmit antennas, we get the familiar signal model
\begin{equation}
\label{eq:signal}
\by
	= \insum_{k=1}^{K} \bH_{k} \bx_{k} + \bz,
\end{equation}
where:
\begin{enumerate}
\item
$\bx_{k}\in\C^{M_{k}}$ is the message transmitted by user $k\in\play$.
\item
$\by\in\C^{N}$ denotes the aggregate signal at the receiver.
\item
$\bH_{k}\in\C^{N\times M_{k}}$ is the $N\times M_{k}$ channel matrix of user $k$.
\item
$\bz\in\C^{N}$ is the ambient noise in the channel,
including thermal, atmospheric and other peripheral interference effects (and modeled for simplicity as a zero-mean, complex Gaussian vector with unit covariance).%
\footnote{Obviously, depending on the structure and statistical properties of the channel matrices $\bH_{k}$,
the signal model \eqref{eq:signal} could be applied to a wide variety of telecommunications systems, ranging from \ac{DSL} uplink networks with T\oe plitz circulant $\bH_{k}$ to \ac{CDMA} and/or \ac{FDMA} radio networks \cite{SCS99}.
For concreteness however, we will interpret \eqref{eq:signal} as an ad hoc multi-user \ac{MIMO} \acl{MAC} with $\bH_{k}$ representing the channel of each link.}
\end{enumerate}

In this context, the average transmit power of user $k$ is simply
\begin{equation}
\label{eq:power}
p_{k}
	= \ex\big[ \|\bx_{k}\|^{2} \big]
	= \tr(\bQ_{k}),
\end{equation}
where $\bQ_{k}$ denotes the user's signal covariance matrix:
\begin{equation}
\label{eq:Qdef}
\bQ_{k}
	= \ex\big[\bx_{k} \bx_{k}^{\dag}\big],
\end{equation}
and the expectation is taken over the Gaussian codebook of user $k$.
Hence, assuming that each user's maximum transmit power is finite, we obtain the feasibility constraints:
\begin{equation}
\label{eq:constraints0}
\bQ_{k} \mgeq0
	\quad
	\text{and}
	\quad
	\tr(\bQ_{k}) \leq P_{k},
\end{equation}
where $P_{k}>0$ denotes the maximum transmit power of user $k$.


The first part of our analysis focuses on \emph{static channels}, i.e. the channel matrices $\bH_{k}$ will be assumed to remain constant (or nearly constant) throughout the transmission horizon (fast-fading channels will be treated in Section \ref{sec:fading}).
In this case, assuming \ac{SUD} at the receiver (i.e. interference by all other users is treated as additive colored noise), each user's achievable transmission rate will be given by the familiar expression \cite{Tel99}:
\begin{flalign}
\label{eq:pay}
\pay_{k}(\bQ)
	& = \log\det\left(\bW_{-k} + \bH_{k} \bQ_{k} \bH_{k}^{\dag}\right) - \log\det\left(\bW_{-k}\right),
\end{flalign}
where $\bQ = (\bQ_{1},\dotsc,\bQ_{K})$
and
\begin{equation}
\label{eq:Wkdef}
\txs
\bW_{-k}
	= \bI + \insum_{\ell\neq k} \bH_{\ell} \bQ_{\ell} \bH_{\ell}^{\dag}
\end{equation}
represents the \acf{MUI} covariance matrix of user $k$.
Thus, given that each user needs to saturate his power constraints in order to maximize his individual rate, we will say that a transmit profile $\eq = (\eq_{1},\dotsc,\eq_{K})$ is at \emph{Nash equilibrium} when no user can unilaterally improve his individual achievable rate $\pay_{k}$, i.e.
\begin{equation}
\label{eq:Nash}
\tag{NE}
\pay_{k}(\eq)
	\geq
	\pay_{k}(\bQ_{k};\eq_{-k})
	\quad
	\text{for all $\bQ_{k}\in\strat_{k}$, $k\in\play$,}
\end{equation}
where $(\bQ_{k};\eq_{-k})$ is shorthand for $(\eq_{1},\dotsc,\bQ_{k},\dotsc,\eq_{K})$ and
\begin{equation}
\label{eq:strat}
\strat_{k}
	= \big\{\bQ_{k}\in\C^{M_{k}\times M_{k}}: \bQ_{k}\mgeq0, \tr(\bQ_{k}) = P_{k} \big\}
\end{equation}
denotes the (compact, convex) set of feasible signal covariance matrices for user $k$.%
\footnote{Under an energy-aware objective, users would not need to saturate their power constraints, but such considerations lie beyond the scope of this paper.}

Dually to the above, if the receiver employs \ac{SIC} techniques to decode the received messages, the users' achievable sum rate will be \cite{YRBC04}:
\begin{equation}
\label{eq:rate}
\txs
\rate(\bQ)
	= \log\det\left( \bI + \insum_{k} \bH_{k} \bQ_{k} \bH_{k}^{\dag}\right).
\end{equation}
In this way, we obtain the sum rate maximization problem:
\begin{equation}
\label{eq:RM}
\tag{RM}
\begin{aligned}
\text{maximize}
	&\quad
	\rate(\bQ),
	\\
\text{subject to}
	&\quad
	\bQ_{k}\in\strat_{k},
	\;
	k =1,\dotsc,K.
\end{aligned}
\end{equation}
Importantly, as can be easily checked, the users' sum rate \eqref{eq:rate} is a \emph{potential function} for the game defined by \eqref{eq:pay} in the sense that
\begin{equation}
\label{eq:potential}
\pay_{k}(\bQ_{k};\bQ_{-k}) - \pay_{k}(\bQ_{k}';\bQ_{-k})
	= \rate(\bQ_{k};\bQ_{-k}) - \rate(\bQ_{k}';\bQ_{-k}).
\end{equation}
Hence, with $\rate$ concave, it follows that the solutions of the Nash equilibrium problem \eqref{eq:Nash} coincide with the solutions of \eqref{eq:RM};
put differently, \emph{optimizing the users' achievable sum rate \eqref{eq:rate} under \ac{SIC} is equivalent to equilibrating the users' individual transmission rates \eqref{eq:pay} under \ac{SUD}}.

For concreteness, in the rest of this paper, we will focus on the sum rate maximization problem \eqref{eq:RM};
however, owing to the above observation, our results will obviously apply to the unilateral equilibration problem \eqref{eq:Nash} as well.

%% file: Learning.tex

The sum rate maximization problem \eqref{eq:RM} is traditionally solved by \acf{WF} methods \cite{CV93}, either iterative \cite{YRBC04,SPB09-sp} or simultaneous \cite{SPB06}.
More precisely, transmitters are typically assumed to have perfect knowledge of the channel matrices $\bH_{k}$ and the aggregate signal-plus-noise covariance matrix
\begin{equation}
\label{eq:W}
\txs
\bW
	= \ex[\by\by^{\dag}]
	= \bI + \insum_{\ell} \bH_{\ell} \bQ_{\ell} \bH_{\ell}^{\dag},
\end{equation}
which is in turn used to calculate the \ac{MUI} covariance matrices $\bW_{-k} = \bW - \bH_{k} \bQ_{k} \bH_{k}^{\dag}$ and ``water-fill'' the effective channel matrices $\wilde \bH_{k} = \bW_{-k}^{-1/2} \bH_{k}$ at the transmitter \cite{YRBC04}.
At a multi-user level, this water-filling process could take place either iteratively (with users updating their covariance matrices in a round robin fashion) \cite{YRBC04} or simultaneously (with all users updating at once) \cite{SPB06}:
the former updating scheme converges always (but slowly for large numbers of users) \cite{YRBC04}, whereas the latter is much faster \cite{SPB06} but may occasionally fail to converge, even in simple, $2$-user parallel \aclp{MAC} \cite{MBML12}.

An added complication in the use of \ac{WF} methods is that they rely on perfect \acf{CSIT} and accurate measurements of $\bW$ at the receiver (who can broadcast this information via a dedicated radio channel or as part of the \ac{TDD} downlink phase).
When such measurements are not available, it is not known whether \ac{WF} methods converge;
on that account, our goal in this section will be to describe a distributed learning method that allows users to attain the system's sum capacity in a distributed way, using only imperfect (and possibly delayed) information.

\subsection{Matrix exponential learning}
\label{sec:MXL}

Instead of relying on fixed-point methods, our approach will rely on tracking the direction of steepest ascent of the system's sum rate in a dual, unconstrained space, and then map the result back to the problem's feasible space via matrix exponentiation.
Formally, assuming for the moment perfect \ac{CSIT}, we will consider the \ac{MXL} scheme:
\begin{equation}
\label{eq:MXL}
\tag{MXL}
\begin{aligned}
\bY_{k}(n+1)
	&= \bY_{k}(n) + \step_{n} \bV_{k}(\bQ(n)),
	\\
\bQ_{k}(n+1)
	&= P_{k} \frac{\exp(\bY_{k}(n+1))}{\tr\big[\exp(\bY_{k}(n+1))\big]},
\end{aligned}
\end{equation}
where:
\begin{enumerate}

\item
$\bV_{k}\equiv\bV_{k}(\bQ)$ denotes the (matrix) derivative of the system's sum rate with respect to each user's covariance matrix:
\begin{equation}
\label{eq:V}
\bV_{k}(\bQ)
	\equiv \nabla_{\bQ_{k}} \rate
	= \nabla_{\bQ_{k}} \pay_{k}
	= \bH_{k}^{\dag} \bW^{-1} \bH_{k}
\end{equation}

\item
$\bY_{k}$ is an auxiliary ``scoring'' matrix which tracks the direction of steepest sum rate ascent.

\item
$\step_{n}$ is a decreasing step-size sequence (typically, $\step_{n} = 1/n$).
\end{enumerate}

\begin{remark}
Given that $\rate(\bQ)$ is real-valued, $\bV_{k}$ and $\bY_{k}$ are automatically Hermitian, so the matrix exponential $\exp(\bY_{k})$ is positive-definite, as required;
the trace normalization in \eqref{eq:MXL} then ensures that $\bQ_{k}(n)$ satisfies the feasibility constraints \eqref{eq:strat} of \eqref{eq:RM} for all $n\geq1$.
In this way, \eqref{eq:MXL} can be seen as
a two-stage, ``primal-dual'' gradient method \cite{Nes09} which reinforces the spatial directions that lead to higher sum rates by allocating more power to the corresponding eigen-directions of the users' covariance matrices.
\end{remark}

To account for imperfect \ac{CSIT} and noisy measurements at the receiver, we will assume that the gradient matrices $\bV$ of \eqref{eq:V} are only known up to a noisy estimate $\hat \bV$.
In particular, we envision the following sequence of events:
\begin{enumerate}
\item
At every update period $n=1,2,\dotsc$, each user $k\in\play$ gets an estimate $\hat\bV_{k}(n)$ of the true gradient matrix $\bV_{k}(\bQ(n))$.
\item
Users update their signal covariance matrices according to \eqref{eq:MXL} and the process repeats.
\end{enumerate}
More concretely, this recurring process may be encoded in algorithmic form as follows:
\begin{algorithm}[H]
{%
\sf
\vspace{2pt}
Parameter:
decreasing step-size sequence $\step_{n}$
\\
Initialize:\;
$n \leftarrow 0$;\;
$\bY_{k} \leftarrow 0$;\;
$\bQ_{k} \leftarrow \frac{P_{k}}{M_{k}} \cdot \bI$
\\[2pt]
\Repeat
{%
$n \leftarrow n+1$;
\\[2pt]
\ForEach
{user $k \in \play$}
{%
	receive estimate $\hat\bV_{k}$ of
	$\bV_{k} = \bH_{k}^{\dag} \bW^{-1} \bH_{k}$;
	\\[2pt]
	update score matrix:
	$\bY_{k} \leftarrow \bY_{k} + \step_{n} \hat\bV_{k}$;
	\\[2pt]
	set covariance matrix:
	$\dis \bQ_{k} \leftarrow P_{k}\cdot \exp(\bY_{k}) \big/\tr[\exp(\bY_{k})]$;
} 
\vspace{.5ex}
\kwd{until} termination criterion is reached.} 
} 
\caption{Matrix Exponential Learning (\acs{MXL}).}
\label{alg:MXL}
\end{algorithm}

The \ac{MXL} algorithm above will be the main focus of our paper, so a few remarks are in order:


\paragraph{Implementation}
\label{par:implementation}

From an implementation point of view, \ac{MXL} has the following desirable properties:
\begin{enumerate}[(P1)]

\item
\emph{Distributedness:}
users have the same information requirements as in distributed \acl{WF} \cite{YRBC04,SPB06,SPB09-sp}.

\item
\emph{Robustness:}
the algorithm does not assume perfect \ac{CSIT} or precise signal measurements at the receiver.

\item
\emph{Statelessness:}
users do not need to know the state of the system (e.g. its topology).

\item
\emph{Reinforcement:}
users reinforce the transmit directions that lead to higher transmission rates.

\end{enumerate}

\smallskip
\paragraph{Assumptions on the measurement errors}
\label{par:noise}

Throughout this paper, we will work with the following statistical hypotheses for the noise process $\bZ_{k}(n) = \hat\bV_{k}(n) - \bV_{k}(\bQ(n))$:
\begin{enumerate}
[({H}1)]

\item
\label{hyp:zeromean}
\emph{Unbiasedness:}
\begin{equation}
\tag{H1}
\label{eq:zeromean}
\ex\left[\bZ(n+1) \vargiven \bQ(n)\right]
	= 0.
\end{equation}

\item
\label{hyp:MSE}
\emph{Finite \acf{MSE}:}
\begin{equation}
\tag{H2}
\label{eq:MSE}
\ex\big[ \norm{\bZ(n+1)}^{2} \given \bQ(n) \big]
	\leq \noisedev^{2}
	\quad
	\text{for some $\noisedev>0$}.
\end{equation}
\end{enumerate}
The statistical hypotheses above allow us to account for a very wide range of error processes:
in particular, we will \emph{not} be assuming \ac{iid} errors, or even errors that are a.s. bounded.%
\footnote{This observation is crucial in the context of wireless networks because measurement errors are typically correlated with the state of the system.}
In fact, Hypotheses \hypref{hyp:zeromean} and \hypref{hyp:MSE} simply amount to asking that the gradient estimate $\hat\bV$ be unbiased and bounded in mean square:
\begin{equation}
\label{eq:Vbound}
\txs
\ex\left[ \smallnorm{\hat\bV_{k}(n+1)}^{2} \vargiven \bQ(n) \right]
	\leq V_{k}^{2}
	\quad
	\text{for some $V_{k}>0$,}
\end{equation}
and our convergence results will be stated with only this mild requirement in mind.

That being said, we obtain even sharper results in some cases under the additional hypothesis:
\begin{enumerate}
[({H}1)]
\setcounter{enumi}{2}
\item
\label{hyp:symmetry}
\emph{Conditionally symmetric error distributions:}
the law of $\bZ(n+1)$ given $\bQ(n)$ is symmetric.
\end{enumerate}
Hypothesis \hypref{hyp:symmetry} also applies to the vast majority of centered processes (uniform, Gaussian, Lévy $\alpha$-stable, Laplace, etc.) and we will use it to further refine our convergence results.

\paragraph{The estimation process}
\label{par:estimation}

As stated, Algorithm \ref{alg:MXL} does not detail how the system's users can obtain an unbiased estimate $\hat\bV_{k}$ of $\bV_{k}$ from individual \ac{CSI} observations and measurements of the received signal covariance matrix $\bW$.
To simplify our presentation, we will state our convergence results below under the assumption that
there is an oracle-like mechanism that returns such estimates to the users on request;
the construction of such a mechanism will then be detailed in Section \ref{sec:estimate}.

\begin{remark*}
We should also note here that if the gradient estimates $\hat\bV_{k}$ are not Hermitian, the score matrices $\hat\bY$ will not be Hermitian either so the users' covariance matrices may fail to be positive-definite.
To avoid such complications, it suffices to take $\big[\hat\bV + \hat\bV^{\dag}\big]/2$ of $\hat\bV$ as an estimator for $\bV$;
for simplicity, we will tacitly assume that this hermitization step has already taken place if necessary.
\end{remark*}


\smallskip
\paragraph{Complexity per iteration}
\label{par:complexity}

From a computational standpoint, it is easy to see that the complexity of each iteration of Algorithm~\ref{alg:MXL} is polynomial (with a low degree) in the number of transmit and receive antennas (for calculations at the transmitter and receiver side respectively).
Specifically, the complexity of the required matrix inversion and exponentiation steps is $\bigoh(N^{\omega})$ and $\bigoh(M_{k}^{\omega})$ respectively, where the exponent $\omega$ is as low as $2.373$ if the processing units employ fast Coppersmith\textendash Winograd methods for matrix multiplication \cite{DS13}.%
\footnote{In particular, the complexity of each iteration of Algorithm 1 is that of matrix multiplication.}
The Hermitian structure of $\bW$ can be exploited to reduce the computational cost of each iteration even further, but such issues lie beyond the scope of this paper;
in practice, the number of transmit and receive antennas are physically constrained by the size of the wireless array, so these operations are quite light.

\smallskip
With all this in mind, our main result is as follows:

\begin{theorem}
\label{thm:conv}
Assume that the \ac{MXL} algorithm (Alg.~\ref{alg:MXL}) is run with nonincreasing step-sizes $\step_{n}$ such that $\sum_{n} \step_{n}^{2} < \sum_{n} \step_{n} = \infty$ and noisy measurements $\hat\bV(n)$ satisfying hypotheses \eqref{eq:zeromean} and \eqref{eq:MSE}.
Then, $\bQ(n)$ converges almost surely to $\argmax_{\bQ} \rate(\bQ)$.

\smallskip

More generally, let $\bar\rate_{n} = \sum_{j=1}^{n} \step_{j}\rate_{j} \big/ \sum_{j=1}^{n} \step_{j}$ denote the time average of the users' sum rate with respect to an arbitrary nonincreasing step-size sequence $\step_{n}$.
Then:
\begin{enumerate}
[\indent i\textup)]
\addtolength{\itemsep}{2pt}
\item
\hfill
\makebox[0pt][r]{
\begin{minipage}[b]{.9\columnwidth}
\begin{equation}
\label{eq:conv-mean}
\ex\big[ \bar\rate_{n} \big]
	\geq R_{\max}
	- \eps_{n},
\end{equation}
\end{minipage}
} 

\item
\hfill
\makebox[0pt][r]{
\begin{minipage}[b]{.9\columnwidth}
\begin{equation}
\label{eq:conv-prob}
\txs
\prob\left(\rate_{\max} - \bar\rate_{n} \geq \eps_{n} + z \right)
	= \bigoh\left(\noisedev^{2} z^{-2} t_{n}^{-2} \insum_{j=1}^{n} \step_{j}^{2} \right)
\end{equation}
\end{minipage}
} 
\end{enumerate}
where
$\rate_{\max} = \max_{\bQ} \rate(\bQ)$ is the system's sum capacity,
$t_{n} = \sum_{j=1}^{n} \step_{j}$,
and
\begin{equation}
\label{eq:conv-bound}
\txs
\eps_{n}
	= t_{n}^{-1} \left[ \insum_{k=1}^{K} \log M_{k} + \frac{1}{2} L^{2} \insum_{j=1}^{n} \step_{j}^{2} \right],
\end{equation}
denotes the algorithm's mean performance guarantee at the \mbox{$n$-th} update period
\textup(in the above, $M_{k}$ is the number of transmit antennas of user $k$ and $L^{2} = \sum_{k=1}^{K} P_{k}^{2} V_{k}^{2}$ is a positive constant depending on the users' maximum transmit powers $P_{k}$ and the mean square $V_{k}^{2}$ of their measurements\textup).

\smallskip

Finally, if \hypref{hyp:symmetry} also holds, the concentration of the algorithm around its mean value at a given iteration is exponential:
\begin{equation}
\label{eq:conv-prob-exp}
\txs
-\log \prob\left( \rate_{\max} - \bar\rate_{n} \geq \eps_{n} + z \right)
	= \bigoh\left(t_{n}^{2} z^{2} \noisedev^{-2} \big/ \insum_{j=1}^{n} \step_{j}^{2} \right)
\end{equation}
\end{theorem}

\begin{proof}
See Appendix \ref{app:proofs-basic}.
\end{proof}

In what follows, we discuss some important points regarding Theorem \ref{thm:conv}:

\setcounter{paragraph}{0}

\smallskip
\paragraph{On the choice of $\step_{n}$}
\label{par:step}

The use of a decreasing step-size sequence $\step_{n}$ in \eqref{eq:MXL} might appear counter-intuitive because it implies that new gradients enter the algorithm with \emph{decreasing weights} (after all, intuition suggests that one should put more weight on recent observations rather than older, obsolete ones).
However, if the algorithm has reached a near-optimal point, a constant step size might cause it to overshoot and miss its mark:
this can be seen clearly from the mean error bound \eqref{eq:conv-bound} which does not vanish as $n\to\infty$ for constant step sizes of the form $\step_{n} = \step$.

As a rule of thumb, the use of a (large) constant step size speeds up the algorithm but may also lead to unwanted oscillations towards the end because it does not dissipate measurement noise and discretization errors:
if the system's users seek to eliminate such phenomena, a decreasing step size should be preferred instead.
In fact, to obtain the best of both worlds, we can consider an adaptive schedule where the method's step-size is initially very high (to accelerate the search of the problem's state space), and is then abruptly decreased when oscillations are detected;
we test such schedules in Section \ref{sec:numerics}.


We should also note here that the requirement $\sum_{n} \step_{n}^{2} < \infty$ for almost sure convergence is tied to the finite \ac{MSE} hypothesis \eqref{eq:MSE} and can be relaxed significantly if there are sharper bounds on the central moments of the estimator $\hat\bV$.
For instance, if the error process $\bZ(n)$ has finite $p$-th moments (cf. Hypothesis \eqref{eq:moments} below), the analysis of \cite{Ben99} shows that the summability condition $\sum_{n}\step_{n}^{2}< \infty$ can be replaced by the much lighter requirement $\sum_{n} \step_{n}^{1+p/2} < \infty$ which allows us to use considerably faster step-size sequences.

\smallskip
\paragraph{Convergence rate}
\label{par:convrate}

If users employ a constant step-size sequence $\step_{j} = \step$ for a number of iterations $n$ that is \emph{fixed in advance} (using ``\textsf{for}'' instead of ``\textsf{while}''), an easy calculation shows that the minimum value of the mean guarantee \eqref{eq:conv-bound} is attained for $\step_{j} = \step = L^{-1} \sqrt{2 \sum_{k} \log M_{k} / n}$ and is equal to:%
\footnote{Note also that $\bar\rate_{n}$ is then given by the standard expression $\bar\rate_{n} = n^{-1}\sum_{j=1}^{n} \rate_{j}$.}
\begin{equation}
\eps_{n}
	= L \sqrt{\frac{2 \sum_{k} \log M_{k}}{n}}.
\end{equation}
Put differently, the mean performance guarantee \eqref{eq:conv-bound} with constant step size falls below $\eps>0$ in $n = 2 L^{2} \sum_{k} \log M_{k} / \eps^{2} = \bigoh(K\eps^{-2})$ iterations.

That being said, this guarantee concerns the \emph{empirical average} of the system's sum rate $\bar \rate_{n} = n^{-1} \sum_{j} \rate_{j}$ and not the users' \emph{instantaneous} sum rate $\rate_{n}$.
Empirical averages evolve much slower than actual values and, as we show in Section \ref{sec:numerics}, the users' instantaneous sum rate $\rate_{n}$ increases much faster and converges to the system's sum capacity within a few iterations (even for very large numbers of users and/or antennas per user).

\smallskip

\paragraph{Large deviations and outage probabilities}
Eq.~\eqref{eq:conv-prob} represents the probability of observing sum rates far below the channel's capacity, so it can be interpreted as a measure of the system's outage probability under \eqref{eq:MXL}.
As such, the tail behavior of \eqref{eq:conv-prob} shows that \ac{MXL} hardens considerably around its deterministic limit:
even though measurement errors can become arbitrarily large (note that a finite \ac{MSE} does not imply a.s. bounded errors), the probability of observing sum rates much lower than what is obtainable with perfect gradient measurements decays very fast \textendash\ exponentially so if \hypref{hyp:symmetry} holds.
In particular, for large $n$, the step-size factor $t_{n}^{-2} \insum_{j=1}^{n} \step_{j}^{2}$ which controls the width of non-negligible large deviations (the parameter $z$) in \eqref{eq:conv-prob} is of order $\bigoh(1/n)$ for step-size sequences of the form $\step_{n} \propto n^{-a}$, $a\in(0,1/2)$, and of order $\bigoh(n^{2a-2})$ for $a\in(1/2,1)$.

Finally, we should also note that the symmetry condition \hypref{hyp:symmetry} is not at all necessary to obtain the exponential rate of decay \eqref{eq:conv-prob-exp}.
As we show in Appendix \ref{app:proofs-basic}, the same bound is also obtained under the following slight reinforcement of \hypref{hyp:MSE}:
\begin{enumerate}
[(H2$'$)]
\item
\label{hyp:moments}
\emph{Subexponential error moment growth:}
\begin{equation}
\tag{H2$'$}
\label{eq:moments}
\ex\left[ \smallnorm{\hat\bV(n+1) - \bV(\bQ(n+1))}^{p} \vargiven \bQ(n) \right]
	\leq \frac{p!}{2} \noisedev^{p}
\end{equation}
for some $\noisedev>0$ and for all $p\in\N$.
\end{enumerate}
This hypothesis is also satisfied by a wide range of probability distributions (including all distributions with finite support, Gaussian, sub-Gaussian and sub-exponential tails) and it offers a useful alternative to \hypref{hyp:symmetry} when the only thing that can be estimated is the errors' raw moments.

\smallskip
\paragraph{Links with matrix regularization}
\label{par:review}

The proof of Theorem \ref{thm:conv} relies on stochastic approximation techniques \cite{Ben99} and a deep connection between matrix exponentiation and the von Neumann quantum entropy.
In fact, as we show in the appendix, \eqref{eq:MXL} is closely related to the matrix regularization techniques of \cite{TRW05,KSST12,MB14} for online learning and the mirror descent machinery of \cite{NJLS09,Nes09} for (stochastic) convex programming.
In particular, the ``convergence-in-the-mean'' bound \eqref{eq:conv-mean} is derived in the same way as the corresponding results of \cite{NJLS09,Nes09}, but the techniques developed therein do not suffice for the much stronger almost sure convergence result that we present here.
An in-depth description of mirror descent methods requires the introduction of considerable technical apparatus from convex analysis and lies beyond the scope of this paper;
for a comprehensive account, see instead \cite{NJLS09,Nes09} and references therein.

\subsection{Unbiased gradient estimators}
\label{sec:estimate}

As we mentioned before, the \ac{MXL} algorithm requires users to have access to unbiased estimates $\hat\bV_{k}$ of their individual gradient matrices \eqref{eq:V}.
Our goal in this section will be to describe a process with which the receiver and the transmitters may estimate $\bV$ based at each step on possibly imperfect signal and channel measurements \textendash\ e.g. obtained through the exchange of pilot feedback signals.

The first step will be to estimate the aggregate \emph{signal precision} (inverse covariance) matrix
\begin{equation}
\label{eq:precision}
\txs
\bP
	= \ex[\by \by^{\dag}]^{-1}
	= \bW^{-1}
	= \left( \bI + \insum_{\ell} \bH_{\ell} \bQ_{\ell} \bH_{\ell}^{\dag} \right)^{-1}
\end{equation}
by sampling the signal $\by \in \C^{N}$ at the receiver.
To that end, since the channel is Gaussian, an unbiased estimate for the received signal covariance $\bW = \ex[\by\by^{\dag}]$ can be obtained from a systematically unbiased sample $\{\by_{s}\}_{s=1}^{S}$ of size $S$ by means of the well-known estimator
\(
\hat \bW
	= \frac{1}{S} \insum_{s=1}^{S} \by_{s} \by_{s}^{\dag}
\).%
\footnote{Since the expected value $\ex[\by] = 0$ of $\by$ is known in advance, we do not need to include an $S/(S-1)$ bias correction factor in the estimate of $\bW$.}
However, $\hat\bW^{-1}$ is a \emph{biased} estimate for $\bW^{-1}$, so inverting $\hat\bW$ directly would introduce a systematic error in $\bP$.
Instead, following \cite{And03}, we will consider the bias-adjusted precision matrix estimator:
\begin{equation}
\label{eq:prec-estimate}
\hat\bP
	= \frac{S - N -1}{S} \hat\bW^{-1},
\end{equation}
where $N$ is the number of antennas at the receiver (the dimension of $\by)$ and $\hat\bW = S^{-1} \insum_{s=1}^{S} \by_{s} \by_{s}^{\dag}$ as before.


Consequently, in the absence of perfect \acl{CSIT}, the transmitters must estimate the individual gradient matrices $\bV_{k} = \bH_{k}^{\dag} \bW^{-1} \bH_{k}$ from the broadcast of $\hat\bP$ and using imperfect measurements of their channel matrices $\bH_{k}$.
To that end, if each transmitter takes $S$ independent measurements $\hat\bH_{k,1},\dotsc,\hat\bH_{k,S}$ of his channel matrix (e.g. via independent pilot sampling),
an unbiased Hermitian estimate for $\bV_{k}$ is given by the expression:
\begin{equation}
\label{eq:V-unbiased}
\hat\bV_{k}
	= \frac{1}{S(S-1)}
	\insum_{s\neq s'} \hat\bH_{k,s}^{\dag} \hat\bP  \hat\bH_{k,s'},
\end{equation}
where $\hat\bP$ is the latest estimate of \eqref{eq:prec-estimate} of $\bW^{-1}$ that was broadcast by the receiver.%
\footnote{Obviously, increasing the sample size $S$ decreases the estimator's \ac{MSE} at the cost of computational complexity.}
Indeed,
given that the sampled channel matrix measurements $\hat\bH_{k,s}$ are assumed stochastically independent, we readily obtain:
\begin{flalign}
\ex[\hat\bV_{k}]
	&= \frac{1}{S(S-1)}
	\insum_{s \neq s'} \ex\big[ \hat\bH_{k,s}^{\dag} \hat \bP \hat\bH_{k,s'} \big]
	= \bH_{k}^{\dag} \bW^{-1} \bH_{k},
\end{flalign}
i.e. \eqref{eq:V-unbiased} constitutes an unbiased estimator of $\bV$.

The construction above provides an estimator $\hat\bV$ with $\ex\big[\hat\bV\big] = \bV$, so Assumption \hypref{hyp:zeromean} holds.
As for the variance of $\hat\bV$, \eqref{eq:V-unbiased} can also be used to derive an expression for $\var(\hat\bV)$ in terms of the moments of $\hat\bP$ and $\hat\bH$.
Since the system input and noise are assumed Gaussian, the former are all finite (and Gaussian-distributed) so the finite mean square error hypothesis \hypref{hyp:MSE} boils down to measuring $\bH$ with finite mean squared error \textendash\ a requirement which is easy to achieve.
In fact, if the sample size $S$ is taken large enough, the central limit theorem guarantees significant control on the variance and higher central moments of $\hat\bV$, so even the tighter requirements \eqref{eq:moments} and \hypref{hyp:symmetry} hold.

\begin{remark}
Under \acf{TDD} operation, the estimation process above is greatly facilitated because the estimates of each user's channel matrix $\bH_k$ can be calculated directly at the transmitter via reverse-link pilots.
In the absence of \ac{TDD} (or in the case where the number of antennas at the receiver is massively large, viz. $N^{2} \gg \insum_{k} M_{k}^{2}$), a more efficient solution would be to feed back to each transmitter an unbiased estimate of $\bV_{k}$ that is calculated directly at the receiver.
\end{remark}



\subsection{Asynchronous updates and delays}
\label{sec:AMXL}

Even though the information requirements of \eqref{eq:MXL} are local in nature, the algorithm itself is not fully decentralized because it relies implicitly on a global timer to coordinate the users' update schedule (similarly to \ac{IWF} and \ac{SWF} methods).
To overcome this limitation, we examine here a fully decentralized variant of Algorithm \ref{alg:MXL} where each user updates his signal covariance matrix based on an individual timer and completely independently of other users.

Of course, in this case, the measurements $\hat\bV_{k}$ may suffer from delays and asynchronicities, so the update structure of Algorithm \ref{alg:MXL} must be suitably modified.
To that end, let $n$ denote the $n$-th \emph{overall} update period in the system, let $\play_{n}\subseteq \play$ denote the subset of users who update at this epoch (typically $\abs{\play_{n}}=1$ if users update at random times),
and let $d_{k}(n)$ be the number of periods that have elapsed since the last update of the $k$-th user.
We then obtain the following asynchronous variant of \eqref{eq:MXL}:
\begin{equation}
\label{eq:AMXL}
\tag{MXL-a}
\begin{aligned}
\bY_{k}(n+1)
	&= \bY_{k}(n)
	+ \step_{n_{k}} \, \one(k\in \play_{n}) \cdot \hat{\bV}_{k}(n), \\
\bQ_{k}(n+1)
	&= P_{k} \frac{\exp(\bY_{k}(n+1))}{\tr[\exp(\bY_{k}(n+1))]},
\end{aligned}
\end{equation}
where $n_{k} = \sum_{j=1}^{n} \one(k\in\play_{j})$ denotes the number of updates that have been performed by user $k$ up to epoch $n$ and $\hat\bV_{k}$ is a noisy (and asynchronous) estimate of \eqref{eq:V}:
\begin{equation}
\label{eq:Vest2}
\hat{\bV}_{k}(n)
	= \bH_{k}^\dag \left[ \bI + \insum_{\ell} \bH_{\ell} {\bQ}_{\ell}(n-d_{\ell}(n)) \bH_{\ell}^\dag \right]^{-1} \bH_{k}
	+ \bZ_{k}(n),
\end{equation}
with $\bZ(n)$ satisfying \hypref{hyp:zeromean} and \hypref{hyp:MSE} as before.

By definition, $\bY_{k}(n)$ and $\bQ_{k}(n)$ are updated at the $(n+1)$-th update period if and only if $k\in\play_{n}$, so every user only needs to keep track of his individual update timer.
In this way, we obtain the following decentralized variant of Algorithm \ref{alg:MXL} (shown here for a single, focal transmitter and with step sizes of the form $\step_{n} = \step/n$ for simplicity):

\begin{algorithm}[H]
{%
\sf
\vspace{2pt}
Parameter:
$\step>0$.
\\[2pt]
Initialize:
$n \leftarrow 0$;\;
$\bY \leftarrow 0$;\;
$\bQ \leftarrow \frac{P}{M} \cdot \bI$
\\[2pt]
\Repeat
{%
\ForEach
{UpdateEvent}
{%
	$n \leftarrow n+1$;
	\\[2pt]
	receive estimate $\hat\bV$ of $\bV$;
	\\[2pt]
	update score matrix:
	$\bY \leftarrow \bY + \step/n \cdot \hat \bV$;
	\\[2pt]
	set covariance matrix:
	$\dis\bQ \leftarrow P \exp(\bY) \big/ \tr[ \exp(\bY) ]$;
} 
\vspace{.5ex}
\textbf{until} termination criterion is reached.
} 
} 
\caption{Asynchronous exponential learning (\acs{AMXL}\acused{AMXL}).}
\label{alg:AMXL}
\end{algorithm}


Remarkably, in this asynchronous context (with delayed, imperfect measurements), we still get:
\begin{theorem}
\label{thm:conv-AMXL}
Assume that the users' delay processes $d_{k}(n)$ are bounded \textup(a.s.\textup) and that the set of users $\play_{n}$ that update at step $n$ is a homogeneous recurrent Markov chain \textendash\ i.e. every user updates at a positive rate.
Then, the iterates of Algorithm \ref{alg:AMXL} converge almost surely to $\argmax_{\bQ} \rate(\bQ)$.
\end{theorem}

\begin{proof}
See Appendix \ref{app:proofs-variants}.
\end{proof}

Obviously, Algorithm \ref{alg:AMXL} enjoys the same implementation properties as Algorithm \ref{alg:MXL}, and, in addition:
\begin{enumerate}
[(P1)]
\setcounter{enumi}{4}
\item
\emph{Asynchronicity:}
there is no need for a global update timer to synchronize the network's wireless users.
\end{enumerate}
In particular, the criteria that trigger an \textsf{UpdateEvent} could be completely arbitrary, so \eqref{eq:AMXL} is more suitable for scenarios where there can be no coordination between the transmitters' update periods.
Otherwise, if an \textsf{UpdateEvent} is triggered simultaneously (e.g. if it is triggered by the receiver's broadcasts), Algorithm \ref{alg:AMXL} reduces to synchronous \ac{MXL} (Alg.~\ref{alg:MXL}):
in this case, the algorithm's convergence can be greatly sped up because more users update per period.

\subsection{Eigenvalues, eigenvectors and learning}
\label{sec:EXL}

We close this section by describing the evolution of the users' transmit eigenvectors and eigenvalues under \eqref{eq:MXL} and using this description to propose an alternative implementation of Algorithm \ref{alg:MXL}.
The key ingredient of our analysis is the following proposition:

\begin{proposition}
\label{prop:MXL-eig}
Let $\{q_{k\alpha},\bu_{k\alpha}\}_{\alpha=1}^{M_{k}}$ be a smooth eigen-system for $\bQ_{k}$ and let $V_{\alpha\beta}^{k} \equiv \bu_{k\alpha}^{\dag} \bV_{k} \bu_{k\beta}$.
Then, the iterates of Algorithm \ref{alg:MXL} track the mean dynamics:
\begin{subequations}
\label{eq:MXL-eig}
\begin{flalign}
\label{eq:MXL-eigenvalues}
\dot q_{k\alpha}
	&=\txs q_{k\alpha} \left(V_{\alpha\alpha}^{k} - P_{k}^{-1}\insum_{\beta=1}^{M_{k}} q_{k\beta} V_{\beta\beta}^{k}\right),
	\\
\label{eq:MXL-eigenvectors}
\dot \bu_{k\alpha}
	&= \insum_{\beta\neq\alpha} V_{\beta\alpha}^{k}\, \left(\log q_{k\alpha} - \log q_{k\beta} \right)^{-1} \bu_{k\beta}.
\end{flalign}
\end{subequations}
\end{proposition}

\begin{proof}
See Appendix \ref{app:proofs-variants}.
\end{proof}

\smallskip

The precise sense in which $\bQ(n)$ ``tracks'' the mean dynamics \eqref{eq:MXL-eig} is explained in Appendix \ref{app:proofs-dynamics};
for our purposes, the most important consequence of Proposition \ref{prop:MXL-eig} is that \eqref{eq:MXL-eig} leads to the following variant of Algorithm \ref{alg:MXL}:

\begin{algorithm}[H]
{%
\sf
\vspace{2pt}
Parameter:
decreasing step-size sequence $\step_{n}$
\\
Initialize:
$n \leftarrow 0$;\;
$q_{k\alpha}$;\;
$\bu_{k\alpha}$
\\[2pt]
\Repeat
{%
$n \leftarrow n+1$;
\\
\ForEach
{user $k \in \play$}
{%
	measure
	$\bV_{k}$;
	\\[2pt]
	update eigenvalues:
	$q_{k\alpha}
	\leftarrow
	q_{k\alpha} + \step_{n} q_{k\alpha} \left(
	V_{\alpha\alpha}^{k} - P_{k}^{-1}\insum_{\beta=1}^{M_{k}} q_{k\beta} V_{\beta\beta}^{k}
	\right)$;
	\\[2pt]
	update eigenvectors:
	$\bu_{k\alpha}
	\leftarrow
	\bu_{k\alpha} + \step_{n}\,\insum_{\beta\neq\alpha} V_{\beta\alpha}^{k} (\log q_{k\alpha} - \log q_{k\beta})^{-1} 	\bu_{k\beta}$;
	\\[2pt]
	correct roundoff errors:
	$\bu \leftarrow \mathtt{Orthonormal}(\bu)$;
	\\[2pt]
	set covariance matrix
	$\bQ_{k} \leftarrow \insum_{\alpha=1}^{M_{k}} q_{k\alpha} \bu_{k\alpha} \bu_{k\alpha}^{\dag}$;
} 
\vspace{.5ex}
\textbf{until} termination criterion is reached.} 
} 
\caption{Eigen-based exponential learning (\acs{EXL}).\acused{EXL}}
\label{alg:EXL}
\end{algorithm}

As in the case of the original \ac{MXL} algorithm, we then obtain:

\begin{theorem}
\label{thm:conv-EXL}
Assume that Algorithm \ref{alg:EXL} is run with sufficiently small, nonincreasing step-sizes $\step_{n}$ such that $\sum_{n} \step_{n}^{2} < \sum_{n} \step_{n} = \infty$.
Then, $\bQ(n)$ converges almost surely to $\argmax_{\bQ} \rate(\bQ)$.
\end{theorem}

We close this section with a few remarks on Algorithm \ref{alg:EXL}:

\setcounter{paragraph}{0}

\smallskip
\paragraph{The orthonormalization step}
\label{par:orthonormalization}

Even though the eigenvector dynamics \eqref{eq:MXL-eigenvectors} preserve orthonormality, Algorithm \ref{alg:EXL} introduces an $\bigoh(\step_{n}^{2})$ round-off error to orthogonality due to discretization.
The call to $\mathtt{Orthonormal}$ performs a basis orthonormalization and it is intended to correct that error in order to yield a covariance matrix $\bQ_{k}$ that satisfies the feasibility constraints \eqref{eq:strat} of \eqref{eq:RM}:
like matrix exponentiation, orthonormalization has the same complexity as matrix multiplication (fast Coppersmith\textendash Winograd methods \cite{DS13} provide an $\bigoh(M_{k}^{2.373})$ bound), so this does not increase the algorithm's (polynomial) complexity.

\smallskip
\paragraph{Noisy measurements}
\label{par:noise-EXL}

Theorem \ref{thm:conv-EXL} has been stated for simplicity for noiseless measurements.
In the case of noisy measurements, the step-size of the algorithm must be tuned adaptively so that $q_{k\alpha} \geq 0$;
it is not hard to do so by using the random step-size techniques of \cite{Ben99}, but we chose to focus on the noiseless case for presentational clarity.

%% file: Fading.tex

In the presence of fading, the users' channel matrices $\bH_{k}$ evolve stochastically over time at a rate which is much faster than the characteristic length of each transmission block;
as a result, the static sum rate function $\rate$ of \eqref{eq:rate} is no longer relevant. 
In this case,
the users' achievable sum rate for fixed $\bQ$ under fast fading is given by the ergodic average \cite{GV97,Tel99}:
\begin{equation}
\label{eq:ergrate}
\txs
\ergrate(\bQ)
	= \ex_{\bH} \left[ \log\det\left(\bI + \insum_{k} \bH_{k} \bQ_{k} \bH_{k}^{\dag} \right) \right],
\end{equation}
where the expectation is now taken with respect to the law of $\bH$ (assumed here to follow a stationary, ergodic process).
Accordingly, we obtain the ergodic rate maximization problem for fast-fading channels:
\begin{equation}
\label{eq:ERM}
\tag{ERM}
\begin{aligned}
\text{maximize}
	&\quad
	\ergrate(\bQ),
	\\
\text{subject to}
	&\quad
	\bQ_{k}\in\strat_{k},
	\;
	k =1,\dotsc,K,
\end{aligned}
\end{equation}
where the users' feasible sets $\strat_{k}$ are defined as in \eqref{eq:strat}:
$\strat_{k} = \{\bQ_{k}\mgeq0: \tr(\bQ_{k}) = P_{k}\}$.%
\footnote{Perhaps more appropriately for the case of interest, the above expression also holds over the long term for block-fading channels, in which case the channel is essentially fixed over each transmission length during  which the instantaneous rate is used.}

Since expectation preserves convexity, the ergodic rate maximization problem \eqref{eq:ERM} remains concave \textendash\ in fact, it is straightforward to show that \eqref{eq:ERM} is \emph{strictly} concave \cite{BLDH10}.
However, given that the integration over the law of $\bH$ is typically impossible to carry out, calculating the ergodic gradient $\bV_{\erg} = \nabla \ergrate$ of $\ergrate$ is a likewise impractical task.
Thus, instead of relying on intricate analytic calculations (that require substantial computation capabilities and a good deal of knowledge regarding the channels' statistics), we will consider the same sequence of events as in the case of static channels:
\begin{enumerate}
\item
At every update period $n=1,2,\dotsc$, each user $k\in\play$
gets an estimate $\hat\bV_{k}(n)$ of the matrix
\begin{equation}
\label{eq:Verg}
\bV_{k}(n)
	= \bH_{k}^{\dag}(n)
	\left[ \bI + \insum_{\ell} \bH_{\ell}(n) \bQ_{\ell}(n) \bH_{\ell}^{\dag}(n) \right]^{-1}
	\bH_{k}(n),
\end{equation}
where $\bH_{k}(n)$ denotes the instantaneous realization of the channel matrix of user $k$ at period $n$.

\item
Users update their signal covariance matrices according to the recursion \eqref{eq:MXL} and the process repeats until a termination criterion is reached.
\end{enumerate}

Formally, writing $\bZ_{k}(n) = \hat\bV_{k}(n) - \ex\left[\bV_{k}(n)\right]$ for the difference between the users' observed estimate $\hat\bV_{k}(n)$ and the expected value of \eqref{eq:Verg}, we will make the same statistical hypotheses for $\bZ$ as in the static regime \textendash\ though we should note here that fluctuations are now due to \emph{both measurement errors and the channels' inherent variability.}
Quite remarkably, despite the change of objective function, we obtain the following convergence result for fast-fading channels:

\begin{theorem}
\label{thm:conv-erg}
Assume that Algorithm \ref{alg:MXL} is run with nonincreasing step-sizes $\step_{n}$ such that $\sum_{n} \step_{n}^{2} < \sum_{n} \step_{n} = \infty$ and noisy measurements $\hat\bV(n)$ satisfying hypotheses \eqref{eq:zeromean} and \eqref{eq:MSE} with respect to \eqref{eq:Verg}.
Then, $\bQ(n)$ converges almost surely to the solution of the ergodic rate maximization problem \eqref{eq:ERM};
moreover, the conclusions of Theorem \ref{thm:conv} for an arbitrary nonincreasing step-size sequence $\step_{n}$ also hold with the static sum rate $\rate$ replaced by the ergodic sum rate $\ergrate$.
\end{theorem}

\begin{proof}
See Appendix \ref{app:proofs-ergodic}.
\end{proof}

In view of Theorem \ref{thm:conv-erg}, we see that Algorithm \ref{alg:AMXL} enjoys the additional property:
\begin{enumerate}
[(P1)]
\setcounter{enumi}{5}
\item
\emph{Flexibility:}
the \ac{MXL} algorithm can be applied ``as-is'' in both static and fast-fading channels.
\end{enumerate}

In particular, the same convergence rate and large deviation estimates that were derived for static channels in the previous section (cf. the remarks following Theorem \ref{thm:conv}) also carry over to the fast-fading regime \textendash\ as does the analysis of Secs. \ref{sec:estimate}, \ref{sec:AMXL} and \ref{sec:EXL} for the users' measurement process and for the asynchronous and eigen-based variants of \ac{MXL} respectively.
The only difference here is that the variance $\sigma$ that appears e.g. in \eqref{eq:conv-prob} and \eqref{eq:conv-prob-exp} is not only due to imperfections in the estimation process of $\bV$, but also stems from the inherent variability of the system's channels due to fast-fading.
We will explore this issue in the following section.

%% file: Numerics.tex

\begin{figure*}[t]
\centering
\subfigure
[Normalized throughput over time for $K=20$ users.]
{\includegraphics[width=.49\textwidth]{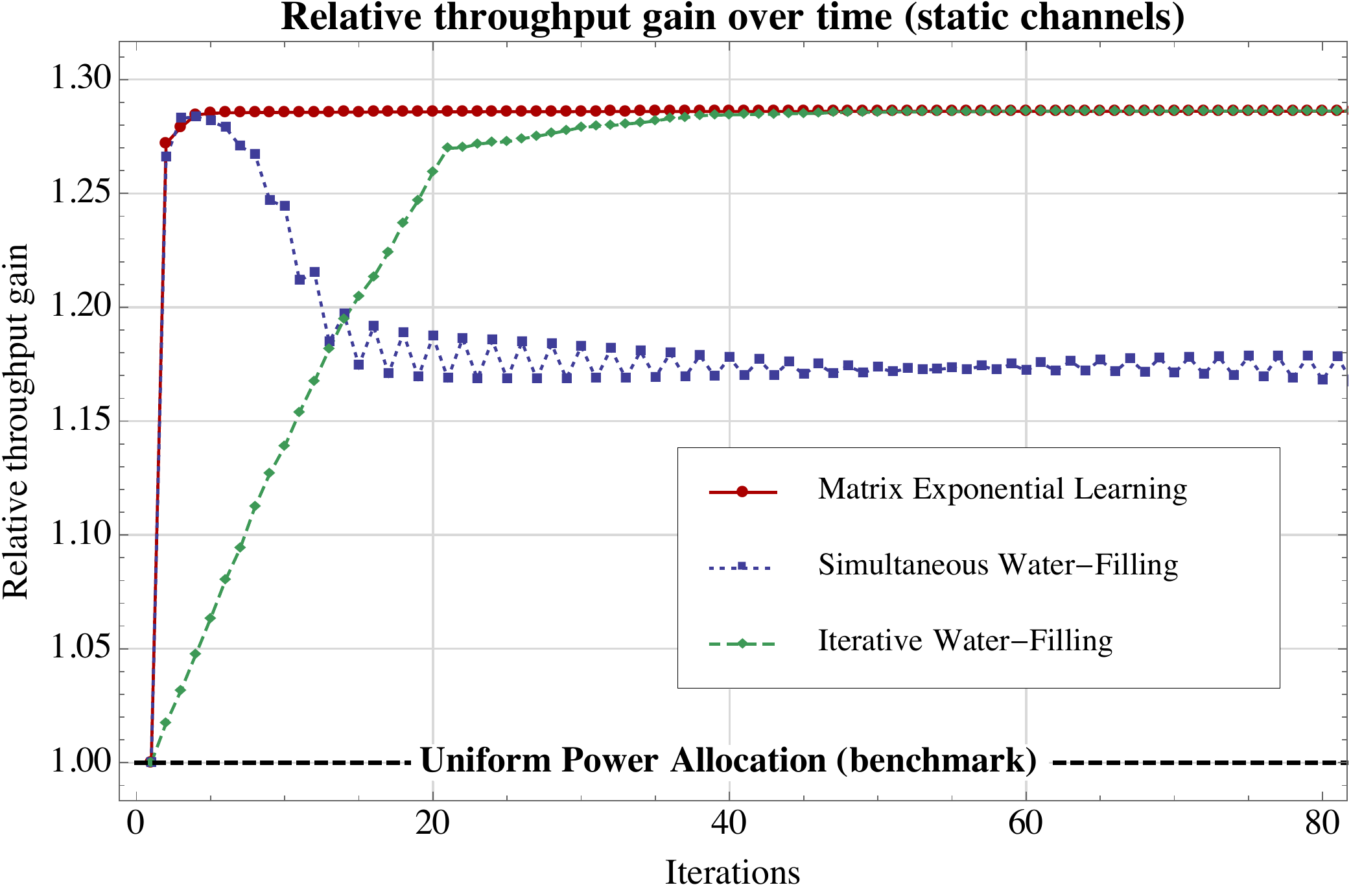}}
\hfill
\subfigure
[Normalized throughput over time for $K=50$ users.]
{\label{fig:regret-rand}
\includegraphics[width=.49\textwidth]{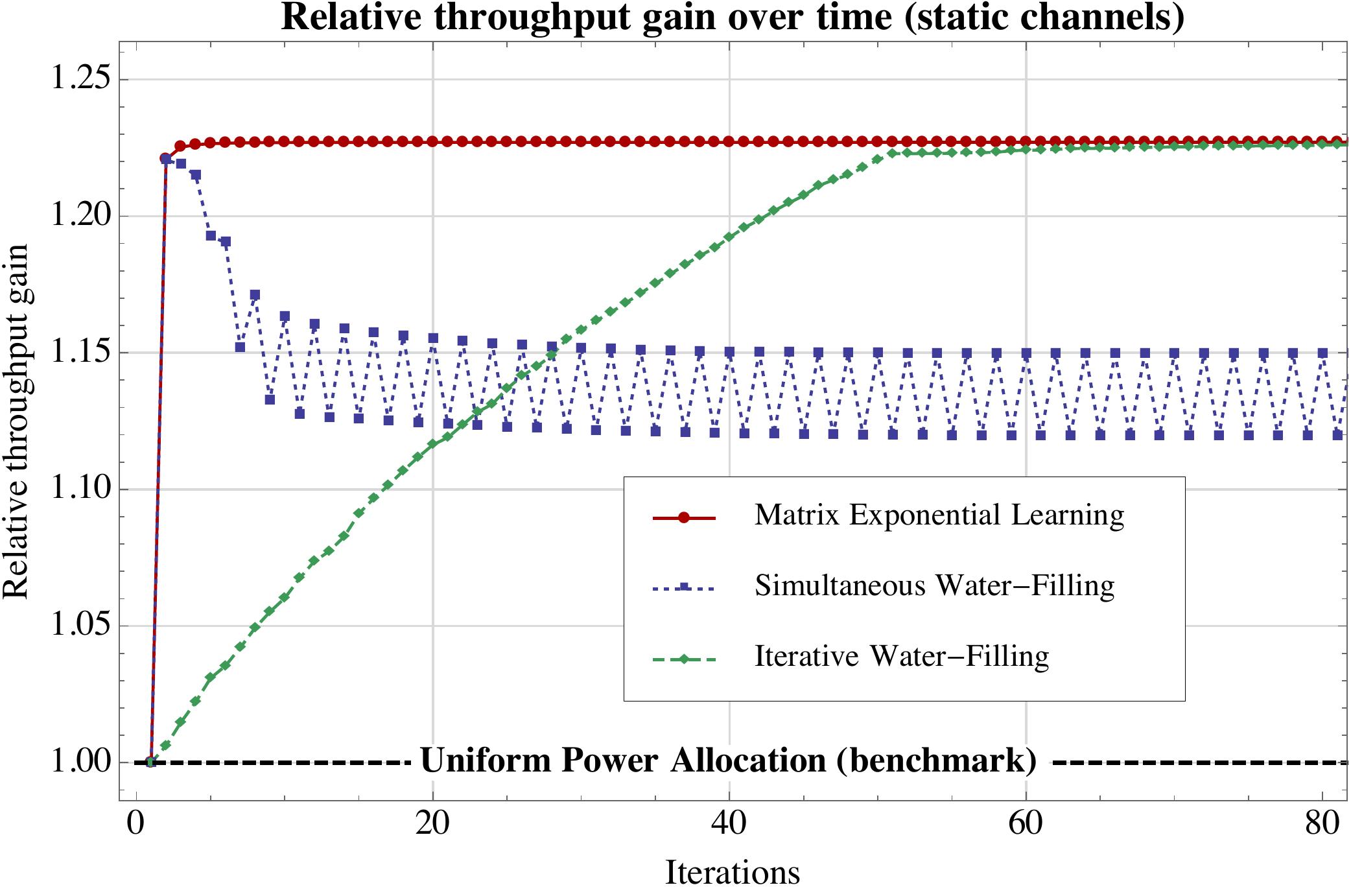}}
\caption{Comparison of \acf{MXL} to \acf{WF} methods.
The classical \ac{IWF} algorithm converges relatively slowly (roughly within $\bigoh(K)$ iterations) because only one user updates per cycle;
the \ac{SWF} variant is much faster (because all users updates simultaneously), but it may fail to converge due to the appearance of best-response cycles in the update process.
By contrast, we see that the \ac{MXL} algorithm converges within a few iterations, even for large numbers of users.}
\label{fig:comp-static}
\end{figure*}

To validate the theoretical analysis of the previous sections and to assess the performance of the \ac{MXL} algorithm (Alg.~\ref{alg:MXL}) in practical scenarios, we conducted extensive numerical simulations from which we illustrate here a selection of the most representative cases.

First, in Figure \ref{fig:comp-static}, we investigate the convergence speed of \ac{MXL} as a function of the number of wireless transmitters, using existing \acf{WF} methods as a benchmark.
To that end, we simulated a multi-user uplink \ac{MIMO} system consisting of a single receiver with $N = 24$ antennas and different numbers of wireless transmitters, each with a random number of transmit antennas, chosen randomly between $2$ and $8$ (due to space limitations, we only present here the case of $K=20$ and $K=50$ users).
The users' channel matrices $\bH_{k}$ were then drawn from a complex Gaussian distribution at the outset of the process and were kept fixed throughout the transmission horizon.%
\footnote{For simplicity, we neglect effects of path-loss in this simulation.
This treatment is reasonable if we assume that users ``invert'' their pathloss function by appropriately compensating their total transmission power level $P_{k}$ \cite{HG08}.}

For comparison purposes, we ran Algorithm \ref{alg:MXL} with constant step size alongside the classical iterative \acl{WF} algorithm proposed in \cite{YRBC04} and the simultaneous variant of \cite{SPB06}, initializing all methods with a uniform power allocation profile that assigns the same power to all transmit antennas (the benchmark case).
The algorithms' performance over time was then assessed by plotting the normalized throughput gain
\begin{equation}
\label{eq:efficiency}
r_{n}
	= \rate_{n} \big/ \rate_{0},
\end{equation}
where $\rate_{n}$ denotes the system's sum rate at the $n$-th iteration of the algorithm,
and the benchmark rate $\rate_{0}$ denotes the system's sum rate when all users employ a uniform beamforming policy.%

As can be seen in Fig.~\ref{fig:comp-static}, the \ac{MXL} algorithm attains the system's sum capacity within a few iterations (effectively, within a \emph{single} iteration for $K=50$ users).%
\footnote{Alternatively, in the game-theoretic context of \eqref{eq:Nash}, this implies that the system's users reach a unilaterally stable Nash equilibrium.}
This convergence behavior represents a marked improvement over traditional \ac{WF} methods, even in moderately-sized systems with $K = 20$ users:
on the one hand, \ac{IWF} is significantly slower than \ac{MXL} (it requires $\bigoh(K)$ iterations to achieve the same performance level as the first iteration of \ac{MXL}), whereas \ac{SWF} may fail to converge altogether due to ``ping-pong'' effects where users overcommit to a given spatial direction by diverting too much power to it all at once (thus ``congesting'' it) and then switch transmit directions in an effort to exploit the ensuing spatial gap (thus relinquishing the benefits from employing it in the first place).%
\footnote{For a detailed account of best-response cycles of this kind, see e.g. \cite{Blu93}.}


\begin{figure*}[t]
\centering
\subfigure
[Learning with a relative error level of 10\%.]
{\includegraphics[width=.49\textwidth]{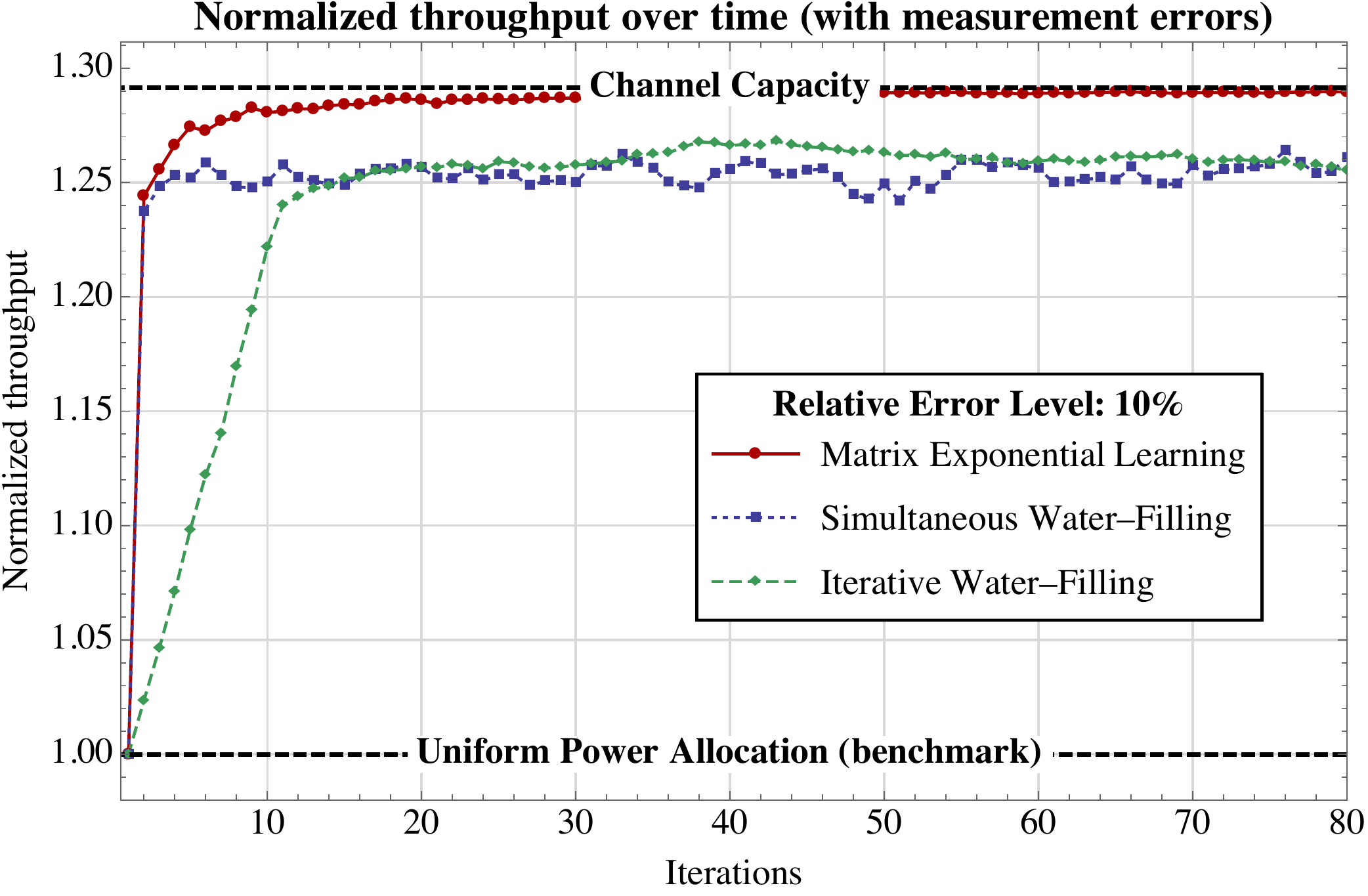}}
\hfill
\subfigure
[Learning with a relative error level of 50\%.]
{\label{fig:regret-rand}
\includegraphics[width=.49\textwidth]{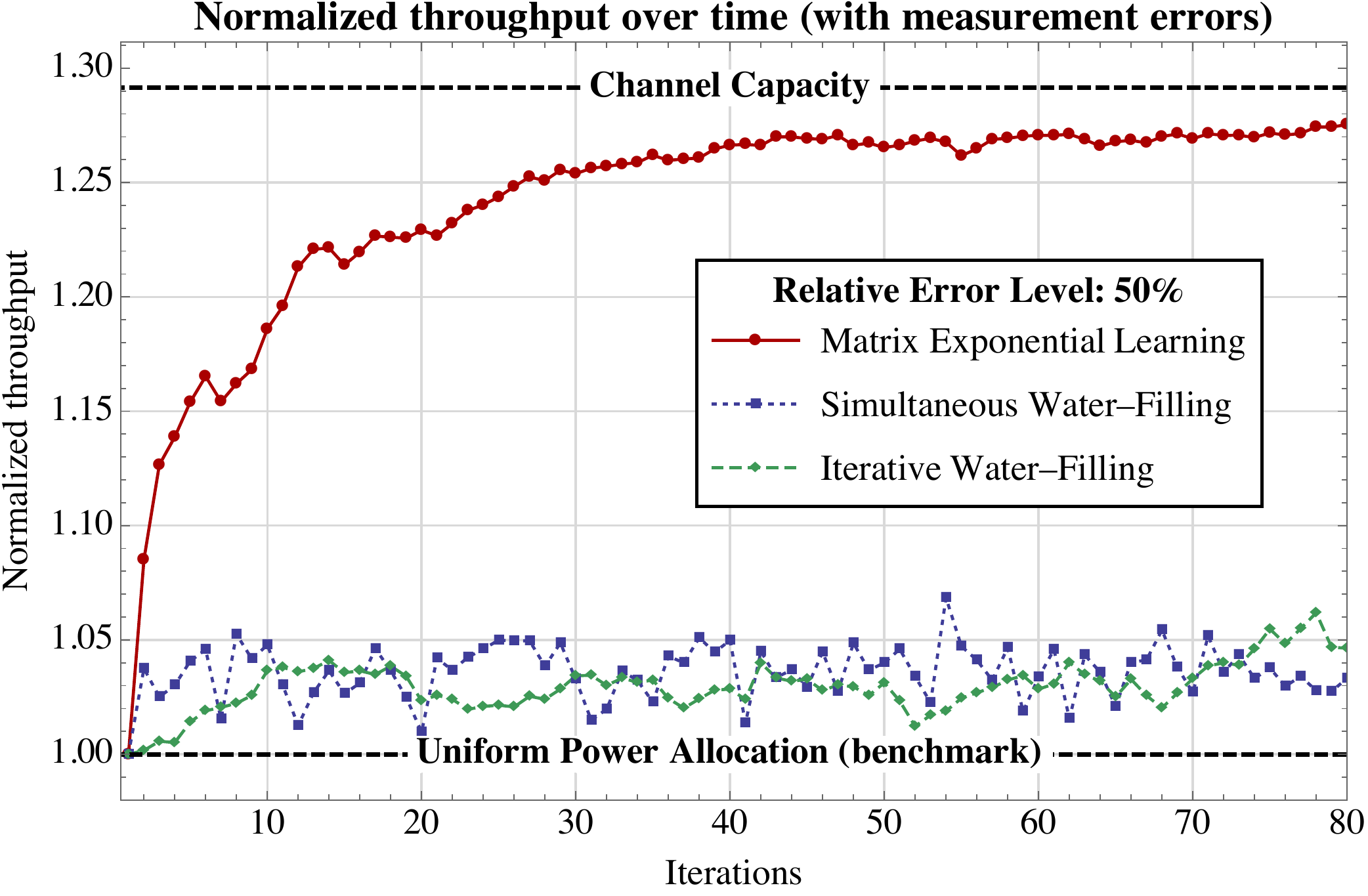}}
\caption{Performance of \acl{MXL} and \acf{WF} methods under imperfect \ac{CSI}.
In contrast to \ac{WF} methods, the \ac{MXL} algorithm attains the channel's sum capacity, even in the presence of very high measurement errors.}
\label{fig:comp-stoch}
\end{figure*}

In Figure \ref{fig:comp-stoch}, we investigate the robustness of \ac{MXL} under imperfect signal and channel measurements, and we compare it to iterative and simultaneous \ac{WF} methods under similar conditions.
Specifically, in Fig.~\ref{fig:comp-stoch}, we simulated a multi-user uplink \ac{MIMO} system consisting of a single receiver with $N = 24$ antennas and $K=20$ wireless transmitters with antenna and channel characteristics as in Fig.~\ref{fig:comp-static}.
To simulate noisy measurements, we used the estimation scheme of Section \ref{sec:estimate} where the received signal precision matrix $\bP = \ex\big[ \by \by^{\dag} \big]^{-1}$ of Eq.~\eqref{eq:precision} is estimated by sampling the aggregate signal $\by$ at the receiver and then feeding the unbiased sample mean \eqref{eq:prec-estimate} to the transmitters.
The measurement noise was controlled by the relative error level of the estimator (deviation/mean), so a relative error level of $\eta$ means that, on average, the estimated matrix lies within $\eta\%$ of its true value (in terms of the Frobenius matrix norm).
We then plotted the efficiency of \ac{MXL} over time for average error levels of $\eta = 10\%$ and $\eta = 50\%$, and we ran the iterative and simultaneous \ac{WF} algorithms with the same sample realizations for comparison.

As can be seen in Figure~\ref{fig:comp-stoch}, the performance of \acl{WF} methods remains acceptable at low error levels, allowing users to attain between 90\% and 95\% of the channel's sum capacity.
However, when the measurement error level gets higher, \acl{WF} (either iterative or simultaneous) offers no perceptible advantage over the users' initial signal covariance matrices (the benchmark case of uniform power allocation across antennas).
By contrast, as predicted by Theorem \ref{thm:conv}, the \ac{MXL} algorithm retains its convergence properties and converges to the system's sum capacity, even under very noisy measurements \textendash\ though, of course, the algorithm's convergence speed is negatively impacted when the measurement noise grows too high.

\begin{figure*}[t]
\centering
\subfigure
[Performance of \ac{MXL} with average user velocity $v=5~\mathrm{m/s}$.]
{\label{fig:tracking-5}
\includegraphics[width=.49\textwidth]{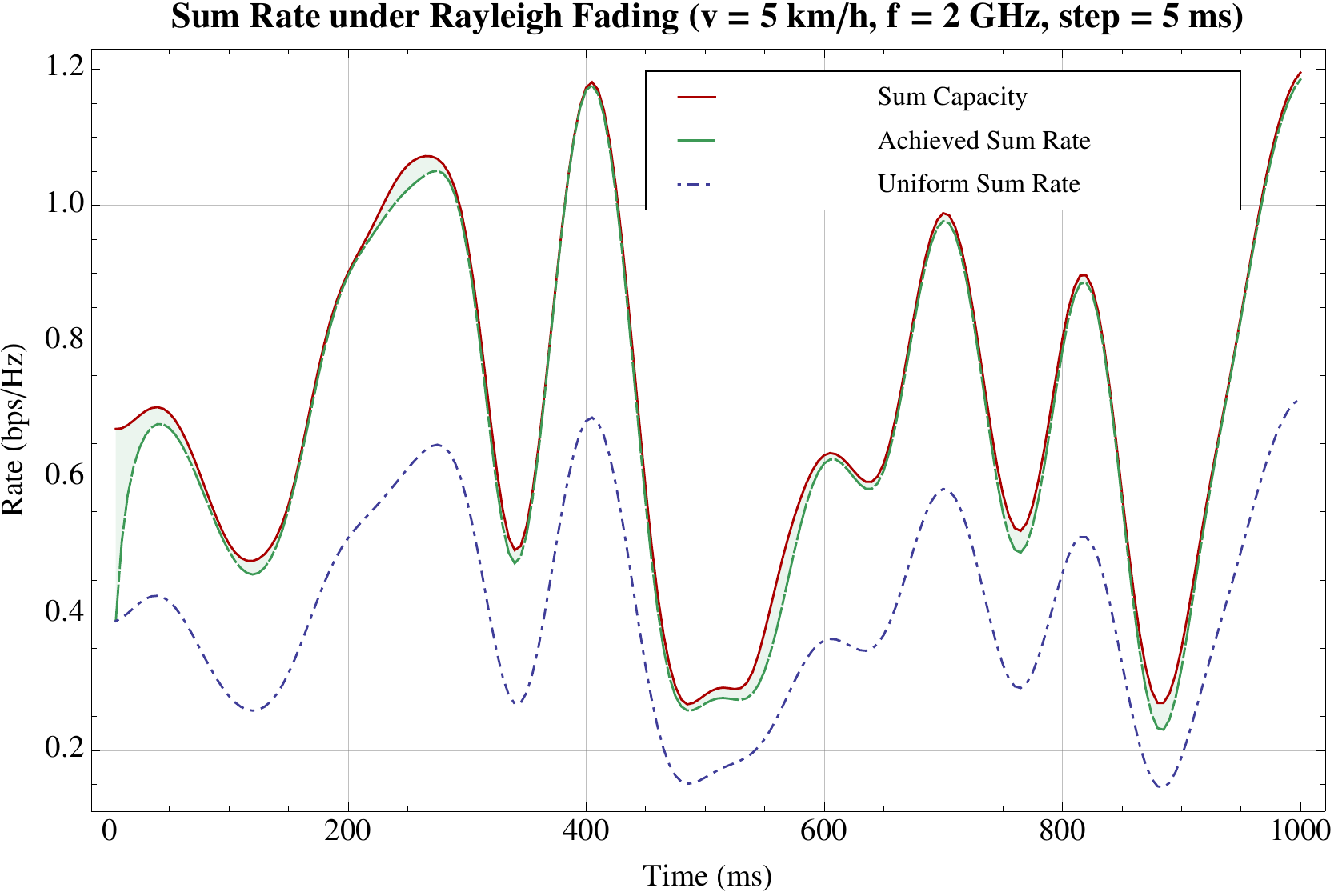}}
\hfill
\subfigure
[Performance of \ac{MXL} with average user velocity $v=15~\mathrm{m/s}$.]
{\label{fig:tracking-15}
\includegraphics[width=.49\textwidth]{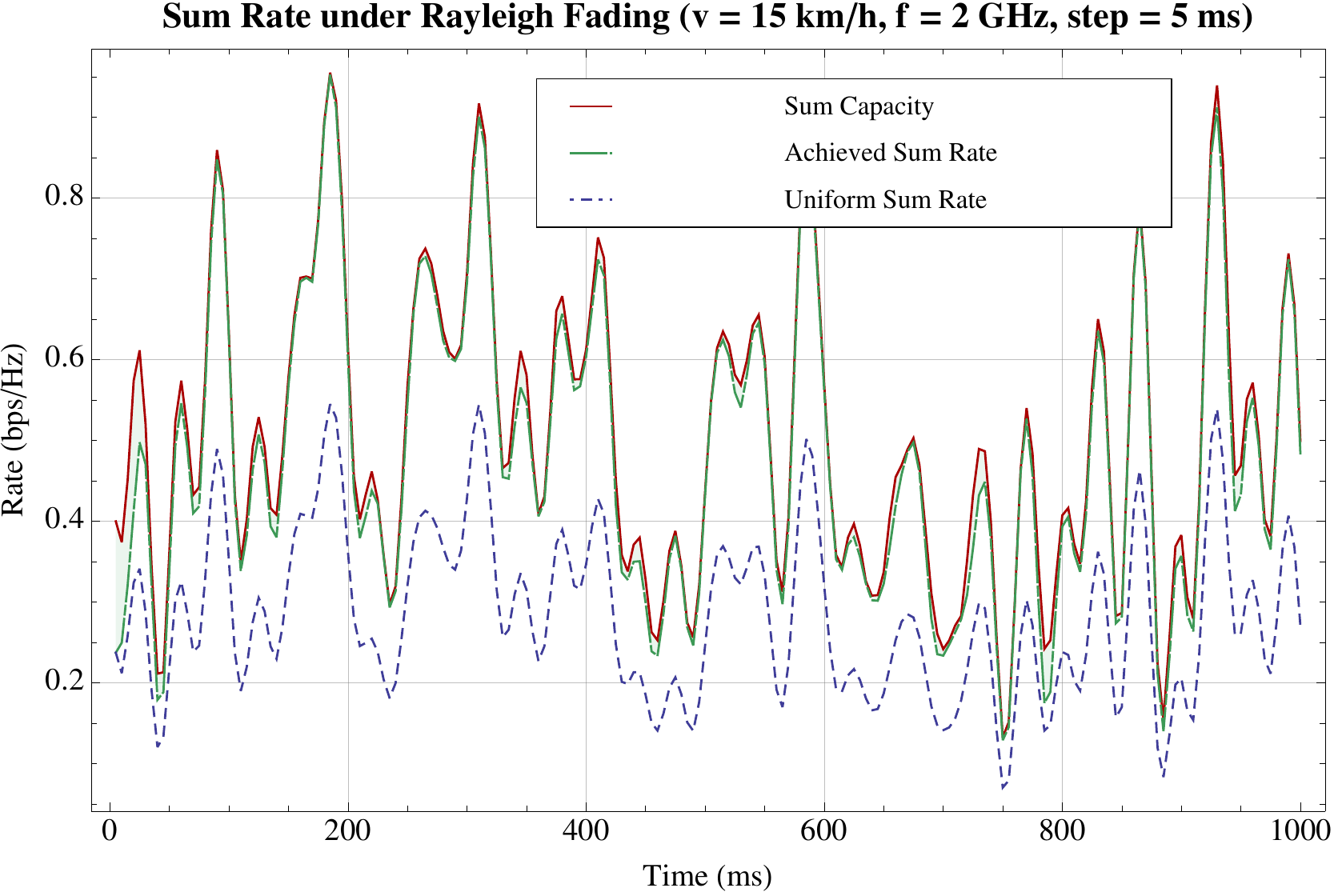}}
\vspace{-1ex}
\caption{Data rates achieved by \ac{MXL} in a dynamic environment with time-varying channels following the Jakes model for Rayleigh fading (model parameters indicated in the figure caption).
The dynamic transmit policy induced by the \ac{MXL} algorithm allows users to track the system's sum capacity remarkably well, even under rapidly changing channel conditions.}
\label{fig:tracking}
\end{figure*}

Finally, to account for time-varying channel conditions, we also plotted the performance of the proposed \ac{MXL} algorithm in non-static channels, following the well-known Jakes model for Rayleigh fading \cite{CCGH+07}.
Specifically, in Figure~\ref{fig:tracking}, we consider a \ac{MIMO} uplink system consisting of a receiver with $N=8$ antennas and $K = 10$ mobile users with $2$ transmit antennas, each transmitting at a central frequency of $f = 2\,\mathrm{GHz}$.
For the users' mobility model, we used the \ac{EPA} and \ac{EVA} models \cite{3GPP} with average user velocities $v = 5~\mathrm{m/s}$ (pedestrian movement) and $v = 15~\mathrm{m/s}$ (vehicular movement in traffic-heavy urban environments).
We then ran the \ac{MXL} algorithm with an update period of $\delta = 5\,\mathrm{ms}$ (one update per frame), and we plotted the algorithm's achieved sum rate $\rate_{n}$ at the $n$-th iteration of the algorithm versus
\begin{inparaenum}
[\itshape a\upshape)]
\item
the system's sum capacity $\rate_{n}^{\max}$ given the current realization of the channel matrices $\bH_{k}(t)$ at time $t = n\delta$;
and
\item
the users' sum rate under uniform power allocation over their antennas.
\end{inparaenum}
Thanks to its high convergence speed and its robustness, the \ac{MXL} algorithm tracks the system's sum capacity remarkably well, despite the channels' variability.
Moreover, the throughput difference between the learned transmit covariance profile and the uniform one shows that this tracking is not an artifact of the system's sum capacity falling within a narrow band of what could be attained by spreading power uniformly over antennas:
users actively track the system's optimum transmit profile as it evolves over time, even under rapidly varying channel conditions.


%% file: Conclusion.tex

In this paper, we introduced a distributed signal covariance optimization algorithm to maximize the uplink sum capacity of multi-antenna users that transmit to a common multi-antenna receiver with only imperfect, possibly delayed and asynchronously updated \acl{CSI} at the users' disposal.
Under fairly mild hypotheses for the moments of the statistics of the estimation imperfections, we showed that the proposed \acf{MXL} algorithm converges rapidly, even for large numbers of users and/or antennas per user;
moreover, the probability that the algorithm deviates beyond a small error from the optimum after a fixed number of iterations is very small (and decays exponentially if the moments of the error process do not grow too fast).
In our view, these robustness properties of \ac{MXL} make it an attractive alternative to \acl{WF} methods that may fail to converge altogether in the presence of measurement noise.
This is confirmed by extensive numerical simulations which exhibit the fast convergence and robustness properties of \ac{MXL} in realistic channel conditions.

We focused on the \ac{MIMO} \acl{MAC} only for simplicity.
The proposed algorithm can be readily extended to a \ac{MIMO}\textendash\acs{OFDM} framework, different precoding schemes (such as MMSE or ZF-type precoders) or to account for other transmission features such as spectral mask constraints, pricing, etc.
The method can also be adapted to wide range channel models (such as the interference channel), where a game-theoretic approach as in \cite{SPB09-sp,SPB08-jsac} is more appropriate:
in this context, a natural question that arises is whether the algorithm converges to a Nash equilibrium, and whether this convergence is retained in the presence of noise.
In fact, thanks to the exponentiation step, our method can be adapted to even more general constrained matrix optimization problems as in \cite{EAS98,AEK08};
we intend to explore these directions in future work.

%% file: App-Proofs.tex

Our goal in this appendix will be to prove the convergence results presented in the rest of our paper.
To that end, we will first establish the convergence of a deterministic, ``mean-field'' dynamical system associated to the \ac{MXL} algorithm, and we will then show that the iterates of \ac{MXL} and its variants comprise a stochastic approximation thereof \cite{Ben99}.
Our robustness results will then follow by exploiting the powerful martingale concentration inequalities of \cite{dlP99}.

For notational clarity, in the rest of this appendix (and unless explicitly stated otherwise), we will treat the case of a single user with maximum transmit power $P=1$;
the general case is simply a matter of taking a direct sum over $k\in\play$ and rescaling by the corresponding maximum power $P_{k}$.

\subsection{The mean dynamics of exponential learning}
\label{app:proofs-dynamics}

We begin by considering the following continuous-time version of the basic recursion \eqref{eq:MXL}:
\begin{equation}
\label{eq:MXL-cont}
\tag{MXL-c}
\begin{aligned}
\dot \bY
	&= \bV,
	\\
\bQ
	&= \frac{\exp(\bY)}{\tr\left[\exp(\bY)\right]},
\end{aligned}
\end{equation}
The following proposition shows that \eqref{eq:MXL-cont} converges to the solution set of the rate maximization problem \eqref{eq:RM}:

\begin{theorem}
\label{thm:conv-cont}
Let $\bQ(t)$ be a solution orbit of \eqref{eq:MXL-cont}.
Then, $\bQ(t)$ converges to a minimizer of \eqref{eq:RM}.
\end{theorem}

A key ingredient in our proof will be the (negative) von Neumann quantum entropy:
\begin{equation}
\label{eq:entropy}
h(\bQ)
	= \tr\left[\bQ \log \bQ\right],
	\quad
	\bQ\in\spectron,
\end{equation}
and its convex conjugate (Legendre transform):%
\footnote{The convexity of $h(\bQ)$ is well known \cite{Ved02}.
To derive the expression \eqref{eq:conjugate} for its conjugate, simply note that $h(\bQ) = \sum_{j} q_{j} \log q_{j}$ where $q_{j}$ are the eigenvalues of $\bQ$;
Eq.~\eqref{eq:conjugate} then follows from the expression $h^{\ast}(y) = \log\sum_{j} \exp(y_{j})$ for the Legendre transform of the (negative) Gibbs entropy.}
\begin{equation}
\label{eq:conjugate}
h^{\ast}(\bY)
	\equiv \max\nolimits_{\bQ\in\spectron} \left\{ \tr[\bY\bQ] - h(\bQ) \right\}
	= \log \tr\left[\exp(\bY)\right],
\end{equation}
where $\spectron$ denotes the (compact) spectrahedron:
\begin{equation}
\label{eq:spectron}
\spectron
	= \left\{\bQ\mgeq0: \tr(\bQ) = 1 \right\}.
\end{equation}
It will also be convenient to introduce the following ``primal-dual'' coupling between $\bQ$ and $\bY$:
\begin{equation}
\label{eq:Fenchel}
\fench(\bQ,\bY)
	= h(\bQ) + h^{\ast}(\bY) - \tr\left[\bQ\bY\right].
\end{equation}
Eq.~\eqref{eq:Fenchel} gathers all the terms of Fenchel's inequality \cite{Roc70}, so, following \cite{MS14}, we will refer to it as the \emph{Fenchel coupling} between $\bQ$ and $\bY$.
Below, we present some key properties of $\fench$:

\begin{lemma}
\label{lem:Fenchel}
With notation as above, we have $\fench(\bQ,\bY) \geq 0$ with equality iff $\bQ = \exp(\bY)/\tr[\exp(\bY)]$.
Moreover:
\begin{equation}
\label{eq:dh*}
\nabla_{\bY} h^{\ast}(\bY)
	= \frac{\exp(\bY)}{\tr\left[\exp(\bY)\right]},
\end{equation}
and
\begin{equation}
\label{eq:dF}
\nabla_{\bY} \fench(\bQ,\bY)
	= \frac{\exp(\bY)}{\tr[\exp(\bY)]} - \bQ.
\end{equation}
\end{lemma}

\begin{IEEEproof}
By standard matrix analysis results \cite{Dat05}, we have:
\begin{flalign}
\nabla_{\bY} h^{\ast}(\bY)
	= \frac{1}{\tr[\exp(\bY)]} \nabla_{\bY} \tr[\exp(\bY)]
	= \frac{\exp(\bY)}{\tr[\exp(\bY)]},
\end{flalign}
and \eqref{eq:dF} follows trivially.
The first part of our claim is then a consequence of the general theory of convex conjugation \textendash\ see e.g. \cite[Chap.~26]{Roc70}.
\end{IEEEproof}

\smallskip

\begin{IEEEproof}[Proof of Theorem \ref{thm:conv-cont}]
Our proof relies on the fact that the Fenchel coupling $H(t) = \fench(\eq,Y(t))$ is a Lyapunov function for \eqref{eq:MXL-cont} for every maximizer $\eq$ of \eqref{eq:RM}.
Indeed, the definition of the Fenchel coupling and Lemma \ref{lem:Fenchel} yield
\begin{equation}
\label{eq:dH}
\dot H
	= \tr\big[ \nabla_{\bY} h^{\ast}(\bY) \cdot \dot\bY \big]
	- \tr\big[ \eq \dot\bY \big]
	= \tr\big[ (\bQ - \eq) \cdot \nabla_{\bQ}\rate \big]
	\leq 0,
\end{equation}
where the inequality in the last step follows from the concavity of $\rate$ and the fact that $\eq$ is a maximizer of $\rate$.
Moreover, equality in \eqref{eq:dH} holds if and only if $\bQ$ is also a maximizer of $\rate$, so $H$ is a Lyapunov function for \eqref{eq:MXL-cont} with respect to $\argmax \rate$.

The above reasoning shows that \eqref{eq:MXL-cont} converges to $\argmax\rate$, but since $\rate$ is not necessarily strictly concave, this does not imply that every trajectory of \eqref{eq:MXL-cont} converges to a specific point in $\argmax\rate$.
To show that this is indeed the case, let $\bQ(t)$ be an orbit of \eqref{eq:MXL-cont} and let $\eq$ be an $\omega$-limit of $\bQ(t)$, i.e. $\bQ(t_{n})\to\eq$ for some increasing sequence $t_{n}\to\infty$.
By Lemma \ref{lem:Fenchel}, this implies that $\fench(\eq,\bY(t_{n}))\to 0$, so, since $\fench(\eq,\bY(t))$ is nonincreasing, we also get $\lim_{t\to\infty} \fench(\eq,\bY(t)) = 0$.
We conclude that $\bQ(t) \to \eq$ (again by Lemma \ref{lem:Fenchel}) and our proof is complete.
\end{IEEEproof}

\subsection{Stochastic approximation and convergence}
\label{app:proofs-basic}

We now proceed to show that the iterates of the \ac{MXL} are asymptotically close to solution segments of \eqref{eq:MXL-cont} of arbitrary length \textendash\ more precisely, that they comprise an \ac{APT} of \eqref{eq:MXL-cont} in the sense of \cite{Ben99}.

\begin{proposition}
Assume that \eqref{eq:MXL} is run with a nonincreasing step-size sequence $\step_{n}$ such that $\sum_{n} \step_{n}^{2} < \sum_{n} \step_{n} = +\infty$ and noisy measurements $\hat\bV_{k}$ satisfying \hypref{hyp:zeromean} and \hypref{hyp:MSE}.
Then, the iterates $\bQ(n)$ of \eqref{eq:MXL} form an \acl{APT} of \eqref{eq:MXL-cont}.
\end{proposition}

\begin{IEEEproof}
Simply note that the recursion \eqref{eq:MXL} can be written in the form:
\begin{equation}
\label{eq:Y0}
\bY(n+1)
	= \bY(n)
	+ \step_{n} \left[ \bV(\bQ(n)) + \bZ(n) \right].
\end{equation}
Since the map $\bY\mapsto\bQ$ is Lipschitz%
\footnote{This follows from the fact that the von Neumann entropy \eqref{eq:entropy} is strongly convex with respect to the nuclear norm \cite{KSST12,Nes09}.}
and the rate function $\rate(\bQ)$ is smooth over the compact spectrahedron $\spectron$, it follows that the map $\bY\mapsto\bV(\bQ(\bY))$ is Lipschitz and bounded.
Our claim then follows from Propositions 4.2 and 4.1 in \cite{Ben99}.
\end{IEEEproof}

With all this said and done, we are finally in a position to prove Theorem \ref{thm:conv}:

\begin{IEEEproof}[Proof of Theorem \ref{thm:conv}]
Let $\eqset = \argmax_{\bQ\in\spectron} \rate(\bQ)$ denote the solution set of \eqref{eq:RM} and assume ad absurdum that $\bQ(n)$ remains a bounded distance away from $\eqset$.
Furthermore, fix some $\eq\in\eqset$ and let $D_{n} = \fench(\bQ,\bY(n))$;
a Taylor expansion of $\fench$ then yields:
\begin{flalign}
\label{eq:Dn0}
D_{n+1}
	&= \fench(\eq,\bY(n+1))
	= \fench(\eq,\bY(n) + \step_{n} \hat\bV(n))
	\notag\\
	&\leq D_{n}
	+ \step_{n} \tr[(\bQ(n) - \eq) \!\cdot\! \bV(\bQ(n))]
	+ \step_{n} \xi_{n}
	+ \tfrac{1}{2} \step_{n}^{2} \smallnorm{\hat\bV(n)}^{2},
\end{flalign}
where $\xi_{n} = \tr[\bZ(n)\cdot (\eq - \bQ(n))]$ and we have used the fact that the convex conjugate $h^{\ast}$ of the von Neumann entropy is $1$-strongly smooth \cite{KSST12}.

Our original assumption that $\bQ(n)$ remains a bounded distance away from $\eqset$ means that $D_{n}$ is bounded away from zero;
moreover, with $\rate$ concave and smooth, we will also have $\tr[\bV(n) \cdot (\bQ(n) - \eq)] \leq -m$ for some $m>0$.
Thus, telescoping \eqref{eq:Dn0} yields:
\begin{equation}
\label{eq:Dn1}
D_{n+1}
	\leq D_{0}
	- t_{n} \left( m - \insum_{j=1}^{n} w_{j,n}\,\xi_{j} \right)
	+ \frac{1}{2} \insum_{j=1}^{n} \step_{j}^{2} \norm{\hat \bV(j)}^{2},
\end{equation}
where $t_{n} = \sum_{j=1}^{n} \step_{j}$ and $w_{j,n} = \step_{j}/t_{n}$.
By the strong law of large numbers for martingale differences \cite[Theorem 2.18]{HH80}, we have $n^{-1} \sum_{j=1}^{n} \xi_{j} \to0$ (a.s.);
hence, with $\step_{n+1}/\step_{n}\leq1$, Hardy's Tauberian summability criterion \cite[p.~58]{Har49} applied to the weight sequence $w_{j,n} = \step_{j}/t_{n}$ yields $\sum_{j=1}^{n} w_{j,n}\, \xi_{j} \to 0$ (a.s.).
Finally, since $\step_{n}$ is square-summable and $\step_{n}\bZ(n)$ is a martingale difference with finite variance, it follows that $\sum_{n=1}^{\infty} \step_{n}^{2} \smallnorm{\hat\bV(n)}^{2} < \infty$ (a.s.) by Theorem 6 in \cite{Cho68}.

Combining all of the above, we see that the RHS of \eqref{eq:Dn1} tends to $-\infty$ (a.s.);
this contradicts the fact that $D_{n}\leq0$, so we conclude that $\bQ(n)$ visits a compact neighborhood of $\eqset$ infinitely often.
Since $\eqset$ attracts any initial condition $\bY(0)$ under the continuous-time dynamics \eqref{eq:MXL}, Theorem 6.10 in \cite{Ben99} shows that $\bQ(n)$ converges to $\eqset$, as claimed.

For the bound \eqref{eq:conv-mean}, note that \eqref{eq:Dn0} can be rewritten as
\begin{equation}
\label{eq:mean1}
\step_{n} \tr[(\eq - \bQ(n))\cdot \bV(\bQ(n))]
	\leq D_{n} - D_{n+1}
	+ \step_{n} \xi_{n}
	+ \frac{1}{2} \step_{n}^{2} \norm{\hat\bV(n)}^{2},
\end{equation}
so, recalling that $\rate$ is concave and $\bV = \nabla_{\bQ}\rate$, we get:
\begin{flalign}
\label{eq:mean2}
\step_{n} \left[ \rate_{\max} - \rate_{n} \right]
	&\leq \step_{n} \tr[(\eq - \bQ(n))\cdot \bV(\bQ(n))]
	\notag\\
	&\leq D_{n} - D_{n+1}
	+ \step_{n} \xi_{n}
	+ \frac{1}{2} \step_{n}^{2} \norm{\hat\bV(n)}^{2}.
\end{flalign}
Thus, taking expectations on both sides and telescoping, we obtain:
\begin{equation}
\label{eq:mean3}
\insum_{j=1}^{n} \step_{j} \left[ \rate_{\max} - \ex[\rate_{j}] \right]
	\leq D_{0}
	+ \frac{1}{2} V^{2} \insum_{j=1}^{n} \step_{j}^{2},
\end{equation}
where we have used the fact that $\ex[\xi_{n}] = 0$ and the finite mean square hypothesis $\ex\big[\smallnorm{\hat\bV(n)}^{2}\big] \leq V^{2}$.
From \eqref{eq:Fenchel}, we have $D_{0} = \fench(\eq,0) \leq \max_{\bQ,\bQ'}\{ h(\bQ) - h(\bQ') \} = \log M$, so \eqref{eq:conv-mean} follows by rearranging \eqref{eq:mean3} and solving for $\ex\big[\bar\rate_{n}\big] = t_{n}^{-1} \sum_{j=1}^{n} \step_{j} \ex\big[\rate_{j}\big]$.

Moreover, for the large deviations bound \eqref{eq:conv-prob}, Eq.~\eqref{eq:mean2} yields $\rate_{\max} - \bar\rate_{n} \leq \eps_{n} + t_{n}^{-1} \insum_{j=1}^{n} \step_{j} \xi_{j}$, so 
\begin{equation}
\label{eq:prob1}
\prob\left( \rate_{\max} - \bar\rate_{n} \geq \eps_{n} + z \right)
	\leq \prob\left( \txs \insum_{j=1}^{n} \abs{\step_{j} \xi_{j}} \geq t_{n} z \right),
\end{equation}
with $\eps_{n} = t_{n}^{-1} \left( \log M + \frac{1}{2} V^{2} \sum_{j=1}^{n} \step_{j}^{2} \right)$ defined as in \eqref{eq:conv-bound}.
Since $\step_{j}\xi_{j}$ is a martingale difference with finite conditional variance $\var(\step_{j}\xi_{j}) \leq A \step_{j}^{2} \noisedev^{2}$ for some $A>0$, Chebyshev's inequality yields:
\begin{equation}
\label{eq:prob2}
\prob\left( \txs \insum_{j=1}^{n} \abs{\step_{j} \xi_{j}} \geq t_{n} z \right)
	\leq \frac{1}{t_{n}^{2} z^{2}} \insum_{j=1}^{n} \var(\step_{j} \xi_{j})
	\leq \frac{A\noisedev^{2}}{t_{n}^{2} z^{2}} \insum_{j=1}^{n} \step_{j}^{2}
	= \bigoh\left( \noisedev^{2} z^{-2} t_{n}^{-2} \insum_{j=1}^{n} \step_{j}^{2} \right).
\end{equation}
Finally, if the conditional distribution of $\hat\bV(n)$ given $\bQ(n-1)$ is symmetric around $\bV(\bQ(n-1))$, the conditional distribution of $\xi_{n}$ will be symmetric around $0$, so the bound \eqref{eq:conv-prob-exp} follows from the exponential concentration inequality (6.1) of \cite{dlP99}.
Otherwise, under the modified hypothesis \eqref{eq:moments}, the exponential bound \eqref{eq:conv-prob-exp} follows from Theorem 1.2A in \cite{dlP99}.
\end{IEEEproof}

\subsection{Variants of \ac{MXL}}
\label{app:proofs-variants}

In this section, we prove the convergence of the variant exponential learning schemes \ac{AMXL} and \ac{EXL} (Algorithms \ref{alg:AMXL} and \ref{alg:EXL} respectively).

\begin{IEEEproof}[Proof of Theorem \ref{thm:conv-AMXL}]
We will show that the recursion \eqref{eq:AMXL} is an asynchronous stochastic approximation of \eqref{eq:MXL-cont} in the sense of \cite[Chap.~7]{Bor08}.
Indeed, by Theorems 2 and 3 in \cite{Bor08}, the recursion \eqref{eq:AMXL} may be viewed as a stochastic approximation of the rate-adjusted dynamics
\begin{equation}
\label{eq:AMXL-cont}
\begin{aligned}
\dot \bY_{k}
	&= \eta_{k} \bV_{k}
	\\
\bQ_{k}
	&= \frac{\exp(\bY_{k})}{\tr[\exp(\bY_{k})]}
\end{aligned}
\end{equation}
where we have momentarily reinstated the user index $k$ and $\eta_{k} = \lim_{n\to\infty} n_{k}/n > 0$ denotes the update rate of user $k$ (the existence and positivity of this limit follows from the ergodicity of the update process $\play_{n}$).
This multiplicative factor does not alter the rest points and \ac{ICT} sets \cite{Ben99} of the dynamics \eqref{eq:MXL-cont}, so \eqref{eq:AMXL-cont} converges to $\argmax\rate$ from any initial condition and the proof of Theorem \ref{thm:conv} carries through essentially verbatim.
\end{IEEEproof}

To prove Theorem \ref{thm:conv-EXL}, we first need to derive the eigen-dynamics \eqref{eq:MXL-eig} induced by \eqref{eq:MXL-cont}:

\begin{proposition}
\label{prop:eigen-dynamics}
Let $\bQ(t)$ be a solution orbit of \eqref{eq:MXL-cont} and let $\{q_{\alpha}(t), \bu_{\alpha}(t)\}$ be a smooth eigen-decomposition of $\bQ(t)$.
Then, $\{q_{\alpha}(t),\bu_{\alpha}(t)\}$ is a solution of the eigen-dynamics \eqref{eq:MXL-eig}.
\end{proposition}

\begin{IEEEproof}
By differentiating the identity $q_{\alpha} \delta_{\alpha\beta} = \bu_{\alpha}^{\dag} \bQ \bu_{\beta}$, we readily obtain:
\begin{flalign}
\label{eq:dQ-eig0}
\dot q_{\alpha} \delta_{\alpha\beta}
	&= \dot\bu_{\alpha}^{\dag} \bQ \bu_{\beta}
	+ \bu_{\alpha}^{\dag} \dot\bQ \bu_{\beta}
	+ \bu_{\alpha}^{\dag} \bQ \dot\bu_{\beta}
	\notag\\
	&= \bu_{\alpha}^{\dag} \dot\bQ \bu_{\beta}
	+ (q_{\alpha} - q_{\beta}) \bu_{\alpha}^{\dag} \dot\bu_{\beta},
\end{flalign}
where the last equality follows by differentiating the orthogonality condition $\bu_{\alpha}^{\dag} \bu_{\beta} = \delta_{\alpha\beta}$.
Thus, by
\begin{inparaenum}
[\itshape a\upshape)]
\item
taking $\alpha=\beta$
and
\item
solving for $\dot\bu_{\alpha}^{\dag}$
\end{inparaenum}
in \eqref{eq:dQ-eig0}, we respectively obtain:
\begin{subequations}
\label{eq:dQ-eig1}
\begin{flalign}
\label{eq:dq}
\dot q_{\alpha}
	&= \bu_{\alpha}^{\dag} \dot \bQ \bu_{\alpha}
	\\
\label{eq:du}
\dot\bu_{\alpha}^{\dag}
	&= \insum_{\beta\neq\alpha} \frac{\bu_{\alpha}^{\dag}\dot\bQ \bu_{\beta}}{q_{\alpha} - q_{\beta}} \bu_{\beta}^{\dag}
\end{flalign}
\end{subequations}

However, by using the Fréchet derivative of the matrix exponential \cite{Wil67}, we readily get:
\begin{flalign}
\label{eq:dQ}
\dot \bQ
	&= \frac{1}{\tr[\exp(\bY)]} \frac{d}{dt} \exp(\bY)
	- \exp(\bY) \frac{\tr[\dot\bY \exp(\bY)]}{\tr[\exp(\bY)]^{2}}
	\notag\\
	&= \frac{1}{\tr[\exp(\bY)]} \int_{0}^{1} \exp((1-s)\bY) \dot \bY \exp(s\bY) \dd s
	- \bQ \tr[\bV \bQ] 
	\notag\\
	&= \int_{0}^{1} \bQ^{1-s} \bV \bQ^{s} \dd s - \bQ \tr[\bV \bQ],
\end{flalign}
and hence:
\begin{flalign}
\label{eq:dQ1}
\bu_{\alpha}^{\dag} \dot\bQ \bu_{\beta}
	&= \int_{0}^{1} \bu_{\alpha}^{\dag} \bQ^{1-s} \bV \bQ^{s}\bu_{\beta} \dd s
	- \tr[\bV \bQ] \cdot \bu_{\alpha}^{\dag} \bQ \bu_{\beta}
	\notag\\
	&= \int_{0}^{1} q_{\alpha}^{1-s} V_{\alpha\beta} q_{\beta}^{s} \dd s
	- q_{\alpha} \delta_{\alpha\beta} \insum_{\gamma} q_{\gamma} V_{\gamma\gamma},
\end{flalign}
where we have set $V_{\alpha\beta} = \bu_{\alpha}^{\dag} \bV \bu_{\beta}$.
Thus, by carrying out the integration in \eqref{eq:dQ1}, we finally obtain:
\begin{equation}
\label{eq:dQ2}
\bu_{\alpha}^{\dag} \dot\bQ \bu_{\beta}
	= \frac{q_{\alpha} - q_{\beta}}{\log q_{\alpha} - \log q_{\beta}} V_{\alpha\beta}
	- q_{\alpha} \delta_{\alpha\beta} \insum_{\gamma} q_{\gamma} V_{\gamma\gamma},
\end{equation}
with the convention $(x-y) / (\log x - \log y) = x$ if $x = y$.
Eq.~\eqref{eq:MXL-eig} then follows by substituting \eqref{eq:dQ2} in \eqref{eq:dQ-eig1}.
\end{IEEEproof}

\begin{IEEEproof}[Proof of Proposition \ref{prop:MXL-eig}]
Combining \eqref{eq:MXL} and the derivative expression \eqref{eq:dQ}, we get:
\begin{flalign}
\label{eq:deltaQ}
\bQ(n+1)
	&= \frac{\exp(\bY(n+1))}{\tr[\exp(\bY(n+1))]}
	= \frac{\exp(\bY(n) + \step_{n} \bV(n))}{\tr[\exp(\bY(n) + \step_{n} \bV(n))]}
	\notag\\
	&= \bQ(n)
	+ \step_{n} \int_{0}^{1} \bQ(n)^{1-s} \bV(n) \bQ(n)^{s} \dd s
	\notag\\
	&- \step_{n} \tr[\bQ(n) \bV(n)] \cdot \bQ(n)
	+ \bigoh\left(\step_{n}^{2} \norm{\bV(n)}^{2}\right),
\end{flalign}
where the term $\bigoh\left(\step_{n} \norm{\bV(n)}^{2}\right)$ is bounded from above by $C\step_{n}^{2}\norm{\bV(n)}^{2}$ for some constant $C$ that does not depend on $\bQ(n)$.
Since $\step_{n}\to0$ by assumption, Remark 4.5 in \cite{Ben99} shows that the quadratic error in \eqref{eq:deltaQ} can be ignored in the long-run, so $\bQ(n)$ is an \ac{APT} of the dynamics \eqref{eq:dQ}.
Hence, by Proposition \ref{prop:eigen-dynamics}, the eigen-decomposition $\{q_{\alpha}(n),\bu_{\alpha}(n)\}$ is an \ac{APT} of \eqref{eq:MXL-eig}, as claimed.
\end{IEEEproof}

\begin{IEEEproof}[Proof of Theorem \ref{thm:conv-EXL}]
Consider the following Euler discretization of the eigen-dynamics \eqref{eq:MXL-eig}:
\begin{subequations}
\begin{flalign}
q_{\alpha}
	&\leftarrow q_{\alpha}
	+ \step_{n} q_{\alpha} \left(V_{\alpha\alpha} - \insum_{\beta} q_{\beta} V_{\beta\beta} \right),
	\\
\bu_{\alpha}
	&\leftarrow \bu_{\alpha}
	+ \step_{n} \insum_{\beta\neq\alpha} \frac{V_{\beta\alpha}}{\log q_{\alpha} - \log q_{\beta}}
	\bu_{\beta},
\end{flalign}
\end{subequations}
i.e. the update step of Alg.~\ref{alg:EXL} without the orthonormalization correction for $\bu_{\alpha}$.
We then obtain:
\begin{flalign}
\bu_{\alpha}^{\dag}(n+1) \cdot \bu_{\beta}(n+1)
	&= \bu_{\alpha}^{\dag}(n) \cdot \bu_{\beta}(n)
	+ \bigoh(\step_{n}^{2}),
\end{flalign}
which shows that the orthonormalization correction in Alg.~\ref{alg:EXL} is quadratic in $\step_{n}$.
Thus, as long as $\step_{n}$ is chosen small enough (so that $q_{\alpha}(n)\geq0$ for all $n$), Remark 4.5 in \cite{Ben99} shows that the iterates of Alg.~\ref{alg:EXL} comprise an \ac{APT} of \eqref{eq:MXL-eig}.
In turn, the same reasoning as in the proof of Prop. \ref{prop:MXL-eig} can be used to show that $\bQ(n) = \insum_{\alpha} q_{\alpha}(n) \bu_{\alpha}(n)\bu_{\alpha}^{\dag}(n)$ is an \ac{APT} of \eqref{eq:MXL-cont}, so $\bQ(n)$ converges to the solution set of \eqref{eq:RM} by Theorem \ref{thm:conv}.
\end{IEEEproof}

\subsection{The fast-fading regime}
\label{app:proofs-ergodic}

\begin{proof}[Proof of Theorem \ref{thm:conv-erg}]
Let $\bV_{\erg} = \nabla \ergrate$ denote the gradient of the ergodic sum rate function $\ergrate$ and consider the dynamics:
\begin{equation}
\label{eq:MXL-erg}
\begin{aligned}
\dot \bY
	&= \bV_{\erg},
	\\
\bQ
	&= \frac{\exp(\bY)}{\tr[\exp(\bY)]}.
\end{aligned}
\end{equation}
The same reasoning as in the proof of Theorem \ref{thm:conv-cont} shows that \eqref{eq:MXL-erg} converges to the unique minimizer of the (strictly concave) sum rate maximization problem \eqref{eq:ERM}.
Moreover, given that $\rate$ is concave for any fixed channel matrix $\bH$ and $\ergrate$ is finite on $\spectron$, we have \cite{Str65}: 
\begin{equation}
\bV_{\erg}
	= \nabla_{\bQ} \ergrate
	= \ex_{\bH} \big[\nabla_{\bQ} \rate(\bQ)\big]
	= \ex_{\bH}[\bV],
\end{equation}
with $\bV$ defined as in \eqref{eq:V}.
With $\bV$ bounded, it follows that $\bV_{\erg}$ is Lipschitz, so Propositions 4.2 and 4.1 in \cite{Ben99} imply that the iterates of \eqref{eq:MXL} with noisy measurements satisfying \hypref{hyp:zeromean} and \hypref{hyp:MSE} comprise a stochastic approximation of the mean dynamics \eqref{eq:MXL-erg}.
The rest of the proof then follows as in the case of Theorem \ref{thm:conv}.
\end{proof}